\documentclass[12pt,a4paper]{article}

\usepackage{amsmath,amsfonts,amssymb,amsthm,mathtools}

\usepackage[english]{babel}
\usepackage[utf8]{inputenc}
\usepackage[T1]{fontenc}

\usepackage{latexsym}
\usepackage[pdftex]{graphicx}
\usepackage{epstopdf}

\usepackage{subfigure}
\usepackage{MnSymbol}

\usepackage{enumitem}

\usepackage{accents}

\usepackage{hyperref}
\usepackage[all]{hypcap}

\usepackage{fancyhdr}
\usepackage[hmarginratio=1:1]{geometry}

\usepackage{authblk}

\usepackage{appendix}

\numberwithin{equation}{section}

\appendixtitleon

\hypersetup{
	pdftitle={},
	pdfauthor={},
	colorlinks=true,
	linkcolor=black,
	citecolor=black,
	filecolor=black,
	urlcolor=blue,
	bookmarksopen,
	bookmarksopenlevel=1,
}

\newcommand{\cE}{\mathcal{E}}

\newcommand{\cH}{\mathcal{H}}
 
\newcommand{\cJ}{\mathcal{J}}

\newcommand{\cO}{\mathcal{O}}
\newcommand{\cP}{\mathcal{P}}

\newcommand{\cS}{\mathcal{S}}
\newcommand{\cT}{\mathcal{T}}

 \let\b=\beta         \let\d=\delta     \let\e=\varepsilon
           \let\ka=\kappa    \let\la=\lambda
\let\m=\mu                          
\let\s=\sigma \let\t=\tau            
   \let\o=\omega     
 \let\D=\Delta

\newcommand{\ee}{\mathrm{e}}
\newcommand{\ii}{\mathrm{i}}
\newcommand{\dd}{\mathrm{d}}

\def\vF{v_{F}}
\def\kF{k_{F}}
\def\vB{v_{B}}
\def\io{\infty}
\def\eps{\epsilon}

\def\sgn{\operatorname{sgn}}
\def\Tr{\operatorname{Tr}}

\def\Map{\operatorname{Map}}

\def\Vect{\operatorname{Vect}}
\def\Diff{\operatorname{Diff}}
\def\wDiff{\widetilde{\operatorname{Diff}}}
\def\vir{\mathcal{V}\mathfrak{ir}}
\def\Id{\operatorname{Id}}
\def\hc{\mathrm{h.c.}}

\newcommand{\pdag}{^{\vphantom{\dagger}}}

\newcommand{\ppr}{^{\vphantom{\prime}}}

\newcommand{\fermionWick}[1]{\left. :\! \hspace{-0.5pt} #1 \hspace{-0.5pt} \!: \right.}

\renewcommand{\Re}{\operatorname{Re}}

\newtheorem{theorem}{Theorem}[section]
\newtheorem{lemma}[theorem]{Lemma}
\newtheorem{proposition}[theorem]{Proposition}

\theoremstyle{definition}

\newtheorem{example}[theorem]{Example}
\theoremstyle{remark}
\newtheorem{remark}[theorem]{Remark}

\begin{document}


\title{%
Inhomogeneous Conformal Field Theory Out of Equilibrium%
}%

\renewcommand\Authfont{\normalsize}
\author[]{%
\vspace{-2pt}%
Per Moosavi%
\thanks{\,\texttt{pmoosavi@phys.ethz.ch}}%
}%

\renewcommand\Affilfont{\footnotesize}
\affil[]{%
\vspace{-8pt}%
Institute for Theoretical Physics, ETH Zurich, Wolfgang-Pauli-Strasse 27, 8093 Z{\"u}rich, Switzerland%
\vspace{-5pt}%
}%

\date{%
\vspace{-10pt}%
\normalsize%
February 23, 2024%
\vspace{-20pt}%
}%


\maketitle


\vspace{-2mm}
\begin{center}
\emph{Dedicated to the memory of Krzysztof Gaw\k{e}dzki.}
\end{center}
\vspace{-2mm}


\begin{abstract}
We study the non-equilibrium dynamics of conformal field theory (CFT) in 1+1 dimensions with a smooth position-dependent velocity $v(x)$ explicitly breaking translation invariance. Such inhomogeneous CFT is argued to effectively describe 1+1-dimensional quantum many-body systems with certain inhomogeneities varying on mesoscopic scales. Both heat and charge transport are studied, where, for concreteness, we suppose that our CFT has a conserved U$(1)$ current. Based on projective unitary representations of diffeomorphisms and smooth maps in Minkowskian CFT, we obtain a recipe for computing the exact non-equilibrium dynamics in inhomogeneous CFT when evolving from initial states defined by smooth inverse-temperature and chemical-potential profiles $\beta(x)$ and $\mu(x)$. Using this recipe, the following exact analytical results are obtained: (i) the full time evolution of densities and currents for heat and charge transport, (ii) correlation functions for components of the energy-momentum tensor and the U$(1)$ current as well as for any primary field, and (iii) the thermal and electrical conductivities. The latter are computed by direct dynamical considerations and alternatively using a Green-Kubo formula. Both give the same explicit expressions for the conductivities, which reveal how inhomogeneous dynamics opens up the possibility for diffusion as well as implies a generalization of the Wiedemann-Franz law to finite times within CFT.
\end{abstract}


\section{Introduction}
\label{Sec:Introduction}


Conformal field theory (CFT) is routinely used to effectively describe universal properties of quantum many-body systems in equilibrium \cite{Cardy:1988}.
Well-known examples include spin chains in the gapless regime at low temperatures and edge currents associated with quantum Hall systems.
The tools of CFT are particularly useful in 1+1 dimensions owing to that the conformal group is infinite dimensional \cite{BPZ}.
Still, it is only recently that this has been used to study collective non-equilibrium properties of 1+1-dimensional quantum many-body systems.

A convenient procedure to theoretically study quantum systems out of equilibrium is to consider the dynamics after a quantum quench.
One such example is the partitioning protocol, where the time evolution is studied starting from an initial state produced by glueing together two semi-infinite systems independently in equilibrium with different thermodynamic variables, such as different temperatures and/or chemical potentials.
This was studied within CFT in, e.g., \cite{CaCa1, BeDo1, BeDo2, BDV, GaTa, CaCa2} among others.
Another example is the smooth-profile protocol used in \cite{LLMM1, LLMM2, GLM, GaKo, SoCa}, where the time evolution is studied starting from initial states defined by smooth inhomogeneous profiles generalizing the usual constant thermodynamic variables.
Recently, there is active interest in extending these kinds of non-equilibrium studies to systems where also the time evolution is inhomogeneous \cite{Kat2, WRL, ADSV, DSVC, DSC, BrDu, WeWu1, RBD, LaMo2, ABF, BCRLM}.

In this paper, we define a family of inhomogeneous models that we refer to as inhomogeneous CFT and study their non-equilibrium properties.
By this, we mean a two-dimensional Minkowskian CFT with spatial translation invariance explicitly broken by replacing the usual constant propagation velocity $v$ by a smooth function $v(x)$ that depends on position $x$.
The Hamiltonian for such a system of finite length $L$ (with periodic boundary conditions) is
\begin{equation}
\label{H_iCFT_introduction}
H
= \int_{-L/2}^{L/2} \dd x\, v(x) [T_{+}(x) + T_{-}(x)],
\end{equation}
where $T_{\pm}(x) = T_{\pm}(x+L)$ are the right- and left-moving components of the energy-momentum tensor in light-cone coordinates (see Sect.~\ref{Sec:Prerequisites} for details) and $v(x) = v(x+L) > 0$.
[Standard CFT is recovered by setting $v(x) = v$.]
Such models have been proposed to effectively describe, for instance, quantum spin chains with certain inhomogeneities varying on mesoscopic length scales, quantum gases in harmonic traps, arctic-circle phenomena, and quantum generalized hydrodynamics \cite{ADSV, DSVC, DSC, RBD, LaMo2, RCDD}.
The first is illustrated in Fig.~\ref{Fig:Inhomogeneous_CFT} for a quantum $XXZ$ spin chain (in the gapless regime and close to half filling%
\footnote{%
But not exactly at half filling (cf.\ Footnote~\ref{Footnote:K_and_Delta}).%
})
with uniformly varying couplings.
Indeed, one can (heuristically) show that this spin chain is effectively described by an inhomogeneous version of the Luttinger model \cite{Tom, Lut, MaLi} with local (point-like) interactions, see, e.g., \cite{Moo}.
The latter will serve as our main example of an inhomogeneous CFT.

\begin{figure}[!htbp]

\centering
\includegraphics[scale=1, trim=14 57 10 20, clip=true]{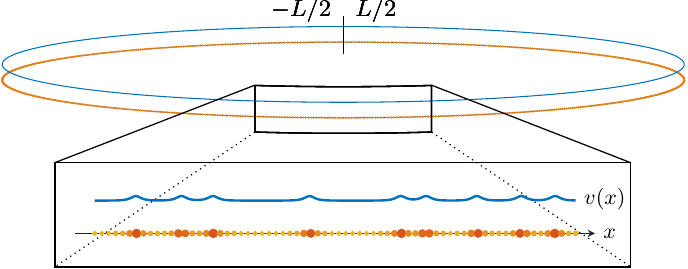}

\caption{%
Illustration of an inhomogeneous CFT with a fixed position-dependent velocity $v(x)$ effectively describing a quantum $XXZ$ spin chain with couplings $J_{j}^{x} = J_{j}^{y} = J_{j}$ and $J_{j}^{z} = J_{j} \D$ (for constant $\D$) between spins on adjacent sites at $x_{j}$ and $x_{j+1}$ uniformly varying on mesoscopic length scales much larger than the lattice spacing but much smaller than the system size.
The spatial dependence of $v(x)$ is directly related to that of the couplings $J_{j}$
(see Section~5.2 in \cite{Moo}), and the color and size of the dots indicate the magnitude of the latter.%
}
\label{Fig:Inhomogeneous_CFT}

\end{figure}

To study both heat and charge transport, for concreteness, we suppose that our CFT has a conserved $\mathrm{U}(1)$ current.
The associated total conserved charge is denoted $Q = \int_{-L/2}^{L/2} \dd x\, [J_{+}(x) + J_{-}(x)]$ where $J_{\pm}(x) = J_{\pm}(x+L)$ are the right- and left-moving components of the $\mathrm{U}(1)$ current in light-cone coordinates (see Sect.~\ref{Sec:Prerequisites} for details).
This is further motivated by that the inhomogeneous local Luttinger model mentioned above is an example of such a CFT.

One purpose of this paper is to lay the mathematical foundations for \cite{LaMo2}, where we studied inhomogeneous CFT with $v(x)$ given by a Gaussian random function.
For such a random CFT, we showed in \cite{LaMo2} that there are both normal and anomalous diffusive contributions to heat transport on top of the usual ballistic one that is the sole contribution in standard CFT.
We mention that the diffusive effect due to the type of randomness in \cite{LaMo2} was recently demonstrated numerically for random quantum spin chains in \cite{AGV} using generalized hydrodynamics \cite{CDY, BCNF}.
This makes clear that the generalization to the inhomogeneous dynamics given by $H$ in \eqref{H_iCFT_introduction} is important as it opens up a mechanism for diffusion within CFT.
Importantly, by generalizing to transport of both heat and charge, we also show that inhomogeneous CFT features a generalization of the Wiedemann-Franz law for finite times.
We mention also that inhomogeneous CFT has close connections to recent works on entanglement Hamiltonians \cite{CaTo}, stochastic CFT \cite{BeDo4, BeDou}, holographic dualities and BTZ black holes \cite{JoLi, MLNR}, and Floquet systems \cite{WeWu2, LCTTNC1, FGVW, HaWe, LapMoo}.

As a final remark, we emphasize that most papers use Euclidean CFT.
One supplementary purpose of this paper is to demonstrate the simplicity and beauty of the Minkowskian theory, which are particularly true when one studies non-equilibrium properties.


\subsection{Projective unitary representations and a non-equilibrium recipe}


As mentioned, we consider CFT in two-dimensional Minkowski space, with the spatial dimension compactified to a circle.
The conformal transformations in this case consist of orientation-preserving diffeomorphisms of the circle.
Recall that the latter form an infinite-dimensional Lie group and that the central extension of the corresponding Lie algebra is the famous Virasoro algebra.
This is important since our use of the full Virasoro algebra makes clear that inhomogeneous CFT contains the recently introduced sine-square-deformed (SSD) CFT as a special case, see, e.g., \cite{Kat2, WRL}, since that only requires the finite-dimensional subalgebra $\mathfrak{sl}(2)$.
For later reference, we introduce this special case as the regularized deformation
\begin{equation}
\label{v_x_gSSD_CFT}
v(x) = v\{ 1 + g[2\cos^2(\pi x/L) - 1] \},
\qquad
v \in \mathbb{R}^{+}, \; g \in [0, 1),
\end{equation}
from which SSD CFT is obtained in the limit $g \to 1^{-}$ (see Remark~\ref{Remark:gSSD_CFT} for further details).%
\footnote{%
The reason for the name becomes apparent by changing coordinate from $x$ to $x - L/2$.%
}
Lastly, we note that our theory is also assumed to have a $\mathfrak{u}(1)$-current algebra, appearing as the central extension of the Lie algebra corresponding to the conserved $\mathrm{U}(1)$ current.

The tools we will present are based on projective unitary representations of the above-mentioned diffeomorphism group and the group of real-valued smooth maps on the circle.
Such methods were used in \cite{GLM} to study the homogeneous time evolution in standard CFT starting from inhomogeneous initial states defined by inverse-temperature and chemical-potential profiles.
The physical setup in \cite{GLM} can be interpreted as a quantum quench from an inhomogeneous system to a homogeneous one.
Here we consider the more general case where both the initial state and the Hamiltonian driving the time evolution are inhomogeneous.

To be more specific, given a smooth inverse-temperature profile $\b(x) = \b(x+L) > 0$ and a smooth chemical-potential profile $\m(x) = \m(x+L)$, let
\begin{equation}
\label{G_iCFT_introduction}
G
= \int_{-L/2}^{L/2} \dd x\, \b(x)
	\Bigl\{
		v(x) \bigl[ T_{+}(x) + T_{-}(x) \bigr]
		- \m(x) \bigl[ J_{+}(x) + J_{-}(x) \bigr]
	\Bigr\}
\end{equation}
be an operator defining a non-equilibrium initial state in the sense that it replaces the combination $\b (H - \m Q)$ with a constant inverse temperature $\b$ and chemical potential $\m$ in the usual Gibbs measure \cite{LLMM2, GLM}.
(We set $\hbar = k_B = 1$ for simplicity.)
We are interested in expectations of the form
\begin{equation}
\label{cO_j_product}
\langle \cO_{1}(t_{1}) \ldots \cO_{n}(t_{n}) \rangle_{\mathrm{neq}}
= \frac{
		\Tr \bigl[
			\ee^{-G} \cO_{1}(t_{1}) \ldots \cO_{n}(t_{n})
		\bigr]
	}{
		\Tr \bigl[ \ee^{-G} \bigr]
	}
\end{equation}
for local operators $\cO_{j}(t_{j}) = \ee^{\ii H t_{j}} \cO_{j} \ee^{-\ii H t_{j}}$ ($j = 1, \ldots, n$) evolving under the inhomogeneous dynamics given by $H$ in \eqref{H_iCFT_introduction}.
Here, locality means that each operator can be expressed as an integral of an operator with finite support (other notions in the literature include quasi- and pseudo-locality), see, e.g., \cite{Doy}.
In principle, these operators can otherwise be arbitrary, but, for simplicity, we will
restrict ourselves to the ``algebra of operators'' generated by the components of the energy-momentum tensor and the $\mathrm{U}(1)$ current together with so-called primary fields with products restricted to non-coincident points in space.

The main proposition (Proposition~\ref{Proposition:Main}) in this paper is a recipe that allows one to compute non-equilibrium expectations of the form in \eqref{cO_j_product} by mapping them to equilibrium ones:%
\begin{equation}
\label{recipe_introduction}
\langle \cO_{1}(t_{1}) \ldots \cO_{n}(t_{n}) \rangle_{\mathrm{neq}}
= \frac{
		\Tr \bigl[
			\ee^{-\b_{0} (H_{0} - \m_{0} Q_{0}) }
			\tilde{\cO}_{1}(t_{1})
			\ldots
			\tilde{\cO}_{n}(t_{n})
		\bigr]
	}{
		\Tr \bigl[ \ee^{-\b_{0} (H_{0} - \m_{0} Q_{0})} \bigr]
	}
\end{equation}
with $H_{0} = \int_{-L/2}^{L/2} \dd x\, v_{0} [T_{+}(x) + T_{-}(x)]$ and $Q_{0} = Q$ for constants $v_{0} > 0$, $\b_{0} > 0$, and $\m_{0} \in \mathbb{R}$ and transformed operators $\tilde{\cO}_{j}(t_{j})$ evolving in a non-trivial way due to the inhomogeneous dynamics.
More precisely, $v_{0}$, $\b_{0}$, $\m_{0}$, and $\tilde{\cO}_{j}$ are given by explicit formulas involving $v(x)$, $\b(x)$, $\m(x)$, and $\cO_{j}$ with the non-trivial time evolution encoded in a natural generalization of the usual light-cone coordinates.
The key to this is to flatten out the profiles and the velocity using diffeomorphisms and smooth maps represented on the Hilbert space of the theory, in generalization of \cite{GLM} for homogeneous dynamics.
This recipe is powerful since, using known results for standard CFT in the literature, one can compute the r.h.s.\ of \eqref{recipe_introduction} in many cases by exact analytical means, not only in the thermodynamic limit $L \to \io$ but also for $L < \io$.
In particular, as $L \to \io$, which will be our focus, the constants $v_{0}$, $\b_{0}$, and $\m_{0}$ will not directly appear in any results.
This agrees with the physical intuition that infinite-volume results for local operators should only depend locally on $v(x)$, $\beta(x)$, and $\mu(x)$ (see Remark~\ref{Remark:v_0_b_0_m_0} for more details).
Note also that \eqref{recipe_introduction} becomes a recipe for computing equilibrium expectations in inhomogeneous CFT by setting $\b(x) = \b$ and $\m(x) = \m$.


\subsection{Summary of results and Wiedemann-Franz law for finite times}


As applications of our non-equilibrium recipe, we derive the following exact analytical results:%
\begin{enumerate}[leftmargin=2.4em, label={(\roman*)}]

\item
\label{Item:i}
The full time evolution of expectations of the form in \eqref{cO_j_product} for the densities and currents associated with heat and charge transport.

\item
\label{Item:ii}
Correlation functions of the form in \eqref{cO_j_product} for components of the energy-momentum tensor and the $\mathrm{U}(1)$ current as well as for any primary fields.
In particular, fully explicit expressions for current-current correlation functions for any inhomogeneous CFT and two-point fermion correlation functions for the inhomogeneous local Luttinger model.

\item
\label{Item:iii}
The thermal and electrical conductivities $\ka_{\mathrm{th}}(\omega)$ and $\s_{\mathrm{el}}(\o)$ as functions of frequency $\omega$.

\end{enumerate}

The conductivities in \ref{Item:iii} are defined as linear-response functions \cite{Kubo} and computed in two ways.
First, dynamically using the explicit expressions for the non-equilibrium expectations for the currents in \ref{Item:i} in the case of kink-like profiles $\b(x)$ and $\m(x)$.
Second, using a Green-Kubo formula for inhomogeneous systems, where the ingredients are the equilibrium current-current correlation functions obtained as special cases of the results in \ref{Item:ii}.
The latter was alluded to but far from properly explained in \cite{LaMo2}, where the explicit expression for $\ka_{\mathrm{th}}(\omega)$ constituted one of two approaches used to show that heat transport acquires diffusive contributions in random CFT.
Using general arguments (see Appendix~\ref{App:Linear-response_theory}), the dynamical and the Green-Kubo approaches must give the same results.
However, the dynamical one turns out to be more direct and makes clear the remarkable role of a quantum anomaly for the final expression for $\ka_{\mathrm{th}}(\omega)$.
This anomaly corresponds to a Schwarzian-derivative term that appears ubiquitously in CFT, but this origin would not be evident from the Green-Kubo approach.

To understand the physical significance of the quantum anomaly (in regard to the conductivities), note that, on general grounds, see, e.g., \cite{GLM, Spo},
\begin{subequations}
\label{ka_th_s_el}
\begin{align}
\Re \ka_{\mathrm{th}}(\o)
& = D_{\mathrm{th}} \pi \d(\o)
		+ \Re \ka_{\mathrm{th}}^{\mathrm{reg}}(\o),
		\label{ka_th} \\
\Re \s_{\mathrm{el}}(\o)
& = D_{\mathrm{el}} \pi \d(\o)
		+ \Re \s_{\mathrm{el}}^{\mathrm{reg}}(\o),
		\label{s_el}
\end{align}
\end{subequations}
where $D_{\mathrm{th}}$ and $D_{\mathrm{el}}$ are the thermal and electrical Drude weights and
$\Re \ka_{\mathrm{th}}^{\mathrm{reg}}(\o)$ and $\Re \s_{\mathrm{el}}^{\mathrm{reg}}(\o)$ are the remaining real regular parts.
The explicit expressions that we derive (see Sect.~\ref{Sec:Applications:Conductivities:Thermal_and_electrical_conductivities}) imply that%
\footnote{%
The second formula assumes that $v(x)$ is not constant since otherwise $\Re \ka_{\mathrm{th}}^{\mathrm{reg}}(\o)$ and $\Re \s_{\mathrm{el}}^{\mathrm{reg}}(\o)$ would be identically zero.%
}
\begin{equation}
\label{Wiedemann-Franz}
\frac{\ka}{c} 
\frac{D_{\mathrm{th}}}
	{D_{\mathrm{el}}}
= \frac{\pi^2}{3 \b},
\qquad
\frac{\ka}{c} 
\frac{\Re \ka_{\mathrm{th}}^{\mathrm{reg}}(\o)}
	{\Re \s_{\mathrm{el}}^{\mathrm{reg}}(\o)}
= \frac{\pi^2}{3 \b}
	\Biggl[ 1 + \biggl( \frac{\o\b}{2\pi} \biggr)^2 \Biggr],
\end{equation}
where $c$ is the central charge appearing in the Virasoro algebra and $\ka$ is the corresponding parameter in the $\mathfrak{u}(1)$-current algebra.
The first formula is essentially the Wiedemann-Franz law, while the second only gives that result in the limit $\o \to 0$.
Indeed, for $\o \neq 0$, there is a correction due to the factor $1 + ( {\o\b}/{2\pi} )^2$, which comes precisely from a Schwarzian derivative involving $v(x)$ and $\b(x)$, and \eqref{Wiedemann-Franz} can be viewed as generalizing the Wiedemann-Franz law within inhomogeneous CFT to finite times.

Likewise, our non-equilibrium recipe allows one to exactly compute alternative conductivities $\ka_{\mathrm{th}}(p, \omega; x')$ and $\s_{\mathrm{el}}(p, \o; x')$ as functions of momentum $p$ and $\omega$ describing the response to perturbations at position $x'$ (see Sect.~\ref{Sec:Applications:Conductivities:Alternative_conductivities}).
For $p = 0$, these satisfy the analogous relations to \eqref{Wiedemann-Franz}, while $p \neq 0$ enters the quantum anomaly, with the effect that the latter appears also in the ratio of the Drude weights and makes the ratio of the real regular parts more involved.


\subsection{Organization of the paper}


In Sect.~\ref{Sec:Prerequisites}, we review well-known facts about Minkowskian CFT that we will need and give examples of such theories.
In Sect.~\ref{Sec:Non-equilibrium_recipe}, we state our main proposition, which gives the recipe behind \eqref{recipe_introduction}.
This recipe is applied in Sect.~\ref{Sec:Applications} to derive the exact analytical results mentioned in \ref{Item:i}\textnormal{--}\ref{Item:iii} above, including \eqref{Wiedemann-Franz} for $\ka_{\mathrm{th}}(\omega)$ and $\s_{\mathrm{el}}(\o)$ as well as the corresponding results for $\ka_{\mathrm{th}}(p, \omega; x')$ and $\s_{\mathrm{el}}(p, \o; x')$.
Our tools are presented in Sect.~\ref{Sec:Main_tools} and used to prove the main proposition.
Concluding remarks are given in Sect.~\ref{Sec:Concluding_remarks}.

Certain topics are deferred to appendices. 
A review of linear-response theory is given in Appendix~\ref{App:Linear-response_theory}, including derivations of a dynamical formula and a Green-Kubo formula for the conductivities.
Appendix~\ref{App:Computational_details} contains computational details for the results in Sect.~\ref{Sec:Applications}.


\section{Prerequisites}
\label{Sec:Prerequisites}


As stressed in Sect.~\ref{Sec:Introduction}, we work in two-dimensional Minkowski space.
Specifically, we let spacetime be the cylinder $\mathbb{R} \times S^1$, where the spatial dimension is the circle $S^1$ of length $L$ parametrized by the coordinate $x \in [-L/2,L/2]$ and time is parameterized by $t \in \mathbb{R}$.
For later reference, we recall that the conformal group in this case is isomorphic to $\Diff_+(S^1) \times \Diff_+(S^1)$, where $\Diff_+(S^1)$ is the group of orientation-preserving diffeomorphisms of the circle \cite{Schott}.
We also introduce its universal covering group $\wDiff_{+}(S^1)$, which consists of all diffeomorphisms $\mathbb{R} \ni x \mapsto f(x) \in \mathbb{R}$ such that
$f(x + L) = f(x) + L$ and $f'(x) > 0$.
In addition, for later reference in regard to gauge transformations, let $\Map(S^1, \mathbb{R})$ denote the group of real-valued smooth maps on the circle, which can be thought of as all smooth functions $\mathbb{R} \ni x \mapsto h(x) \in \mathbb{R}$ such that $h(x + L) = h(x)$.


\subsection{Conformal transformations}
\label{Sec:Prerequisites:Conformal_transformations}


For all our intents and purposes, by a CFT we mean a unitary 1+1-dimensional quantum field theory that is invariant under conformal transformations.
The main objects of such a theory are the $L$-periodic operators $T_{\pm}(x)$ in \eqref{H_iCFT_introduction}.
These are the right- and left-moving components of the energy-momentum tensor in the usual light-cone coordinates $x^{\pm} = x \pm vt$ and satisfy the equal-time commutation relations
\begin{subequations}
\label{Virasoro_alg_pos_space}
\begin{align}
\bigl[ T_{\pm}(x), T_{\pm}(x') \bigr]
& = \mp 2\ii \d'(x-x') T_{\pm}(x')
		\pm \ii \d(x-x') T_{\pm}'(x')
		\pm \frac{c}{24\pi} \ii \d'''(x-x'), \\
\bigl[ T_{\pm}(x), T_{\mp}(x') \bigr]
& = 0,
\end{align}
\end{subequations}
where $c$ is the central charge and $\d(x)$ is the $L$-periodic delta function.
We recall that $T_{+} = T_{--}$, $T_{-} = T_{++}$, and $T_{+-} = 0 = T_{-+}$ in more conventional notation, where pairs of signs refer to the light-cone coordinates, and that $T_{\pm} = T_{\pm}(x^{\mp})$ only depends on one of these coordinates.

In Fourier space, the commutation relations in \eqref{Virasoro_alg_pos_space} correspond to those of two commuting copies of the Virasoro algebra (see Sect.~\ref{Sec:Main_tools}).
For completeness, we recall that the Hilbert space $\cH$ of our theory is a (possibly infinite) direct sum of unitary highest-weight representations of two commuting copies of the Virasoro algebra such that $\cH = \cH_+ \otimes \cH_-$ with $\cH_{+(-)}$ corresponding to right- (left-) moving excitations.

Another important class of operators are Virasoro primary fields.
Recall that a field $\Phi$ is said to be Virasoro primary with conformal weights $(\D^{+}_{\Phi}, \D^{-}_{\Phi})$ if it obeys
\begin{equation}
\label{U_Phi_Ui_prerequisites}
\Phi(x^-, x^+)
\to U(f_+, f_-) \Phi(x^-, x^+) U(f_+, f_-)^{-1}
= f'_+(x^-)^{\D^+_{\Phi}} f'_-(x^+)^{\D^-_{\Phi}}
		\Phi(f_+(x^-), f_-(x^+))
\end{equation}
under conformal transformations given by $f_\pm \in \wDiff_+(S^1)$ unitarily implemented by $U(f_+, f_-) = U_+(f_+) U_-(f_-)$ with $U_\pm(f_\pm)$ acting non-trivially only on $\cH_{\pm}$.
Recall also that the $T_{\pm}$-operators are not Virasoro primary since
\begin{equation}
\label{U_Tpm_Ui_prerequisites}
T_{\pm}(x^{\mp})
\to U(f_+, f_-) T_{\pm}(x^{\mp}) U(f_+, f_-)^{-1}
= f'_\pm(x^{\mp})^2 T_{\pm}(f_\pm(x^{\mp}))
		- \frac{c}{24 \pi} \{ f_\pm(x^{\mp}), x^{\mp} \}
\end{equation}
under conformal transformations, where
\begin{equation}
\label{Schwarzian_derivative}
\{ f(x), x \}
= \frac{f'''(x)}{f'(x)} - \frac{3}{2} \left( \frac{f''(x)}{f'(x)} \right)^2
\end{equation}
is the Schwarzian derivative of $f(x)$, see, e.g., \cite{FMS}.
The latter is an anomaly coming from the Schwinger term (the third term) in \eqref{Virasoro_alg_pos_space}, and it is the reason why the $T_{\pm}$-operators fail to be Virasoro primary with conformal weights
$(\D^{+}_{T_{+}}, \D^{-}_{T_{+}}) = (2,0)$
and
$(\D^{+}_{T_{-}}, \D^{-}_{T_{-}}) = (0,2)$.

Lastly, we recall that $T_{\pm}(x^\mp)$ and $\Phi(x^-, x^+)$ are actually operator-valued distributions.


\subsection{Gauge transformations}
\label{Sec:Prerequisites:Gauge_transformations}


We suppose that our CFT has a conserved $\mathrm{U}(1)$ current and let $J_{\pm}(x)$ denote the right- and left-moving components of this current in light-cone coordinates.
Then, in addition to \eqref{Virasoro_alg_pos_space},
\begin{subequations}
\label{Virasoro_and_u1_current_alg_pos_space}
\begin{align}
\bigl[ J_{\pm}(x), J_{\pm}(x') \bigr]
& = \mp \frac{\ka}{2\pi} \ii \d'(x-x'),
& \bigl[ J_{\pm}(x), J_{\mp}(x') \bigr]
& = 0,
	\label{u1_current_alg_pos_space} \\
\bigl[ T_{\pm}(x), J_{\pm}(x') \bigr]
& = \mp \ii \d'(x-x') J_{\pm}(x') \pm \ii \d(x-x') J_{\pm}'(x'),
& \bigl[ T_{\pm}(x), J_{\mp}(x') \bigr] 
& = 0,
\end{align}
\end{subequations}
where $\ka$ plays a similar role as $c$ in \eqref{Virasoro_alg_pos_space}.
As for $T_{\pm}$, we recall that $J_{\pm} = J_{\pm}(x^{\mp})$ only depends on one of the light-cone coordinates, while, different from \eqref{U_Tpm_Ui_prerequisites},
\begin{equation}
\label{U_Jpm_Ui_prerequisites}
J_{\pm}(x)
\to U(f_+, f_-) J_{\pm}(x) U(f_+, f_-)^{-1}
= f'_\pm(x) J_{\pm}(f_\pm(x))
\end{equation}
under conformal transformations.
Clearly, the $J_{\pm}$-operators are Virasoro primary with conformal weights
$(\D^{+}_{J_+}, \D^{-}_{J_+}) = (1,0)$ and $(\D^{+}_{J_-}, \D^{-}_{J_-}) = (0,1)$.

In Fourier space, the commutation relations in \eqref{u1_current_alg_pos_space} correspond to those of two commuting copies of the $\mathfrak{u}(1)$-current algebra.
The above can be generalized to more complicated current algebras, see, e.g., \cite{KnZa}, but for simplicity we consider only the Abelian case.
(Note that the examples we have in mind are such CFTs, see Sect.~\ref{Sec:Prerequisites:Examples}.)

Gauge transformations can be divided into large and small.
We only consider small $\mathrm{U}(1)$ gauge transformations, which are of the form $\ee^{\ii h}$ with $h \in \Map(S^1, \mathbb{R})$.
Under such transformations,
\begin{subequations}
\label{V_Tpm_Jpm_Vi_prerequisites}
\begin{align}
T_{\pm}(x^{\mp})
& \to V(h_+, h_-) T_{\pm}(x^{\mp}) V(h_+, h_-)^{-1}
= T_{\pm}(x^{\mp}) + h'_\pm(x^{\mp}) J_{\pm}(x^{\mp}) + \frac{\ka}{4\pi}h'_\pm(x^{\mp})^2,
	\label{V_Tpm_Vi_prerequisites} \\
J_{\pm}(x^{\mp})
& \to V(h_+, h_-) J_{\pm}(x^{\mp}) V(h_+, h_-)^{-1}
= J_{\pm}(x^{\mp}) + \frac{\ka}{2\pi}h'_\pm(x^{\mp})
	\label{V_Jpm_Vi_prerequisites}
\end{align}
\end{subequations}
for $h_\pm \in \Map(S^1, \mathbb{R})$ unitarily implemented by $V(h_+, h_-) = V_+(h_+)V_-(h_-)$ with $V_\pm(h_\pm)$ acting non-trivially only on $\cH_{\pm}$.
Similar to \eqref{U_Phi_Ui_prerequisites}, a field $\Phi$ that obeys
\begin{equation}
\label{V_Phi_Vi_prerequisites}
\Phi(x^-, x^+)
\to V(h_+, h_-) \Phi(x^-, x^+) V(h_+, h_-)^{-1}
= \ee^{-\ii [h_+(x^-) \tau^+_{\Phi} - h_-(x^+) \tau^-_{\Phi}]} \Phi(x^-, x^+)
\end{equation}
is said to be $\mathrm{U}(1)$ primary with associated weights $(\tau^+_{\Phi}, \tau^-_{\Phi})$.

We say that a field $\Phi$ is \emph{primary} if it obeys both \eqref{U_Phi_Ui_prerequisites} and \eqref{V_Phi_Vi_prerequisites} under conformal and gauge transformations with weights $(\D^{+}_{\Phi}, \D^{-}_{\Phi})$ and $(\tau^{+}_{\Phi}, \tau^{-}_{\Phi})$.
It follows from \eqref{V_Tpm_Jpm_Vi_prerequisites} that both the $T_{\pm}$- and the $J_{\pm}$-operators fail to be $\mathrm{U}(1)$ primary.

Lastly, we recall that also $J_{\pm}(x^\mp)$ are actually operator-valued distributions.


\subsection{Examples}
\label{Sec:Prerequisites:Examples}


Below we present three examples of CFTs and recall what the $T_{\pm}$- and $J_{\pm}$-operators and the primary fields are in each case.

\begin{example}[Free massless fermions on the circle]
\label{Example:Free_fermions}
Let $\psi^-_{r}(x)$ and $\psi^+_{r}(x) = \psi^-_{r}(x)^\dagger$ for $r = \pm$ be fermionic fields satisfying the usual anti-commutation relations
\begin{equation}
\label{psi_psi_anticomm_rels}
\bigl\{ \psi^-_{r\ppr}(x), \psi^+_{r'}(x') \bigr\}
= \d_{r, r'} \d(x-x'),
\quad
\bigl\{ \psi^\pm_{r\ppr}(x), \psi^\pm_{r'}(x') \bigr\}
= 0
\end{equation}
and anti-periodic boundary conditions $\psi^{\pm}_{r}(x + L) = -\psi^{\pm}_{r}(x)$.
The index $r = +(-)$ denotes right- (left-) moving fermions.
This defines a CFT with $c = 1$ given by the Hamiltonian
\begin{equation}
\label{H_free_fermions}
H_{F}
= \sum_{r=\pm} \int_{-L/2}^{L/2} \dd x\, \frac{1}{2}
	\bigl[
		\! \fermionWick{ \psi^+_{r}(x) \left(-\ii r \vF \partial_x \right) \psi^-_{r}(x) } \!
		+ \hc
	\bigr]
	- \frac{\pi\vF}{6L},
\end{equation}
where $\vF > 0$ is the Fermi velocity and $\fermionWick{\cdots}$ indicates (fermion) Wick ordering with respect to the vacuum (i.e., the filled Dirac sea).
Here,
$T_{\pm}(x)
= [
		\! \fermionWick{ \psi^+_{\pm}(x)(\mp\ii\partial_x) \psi^-_{\pm}(x) } \! + \hc
	]/2
	- {\pi}/{12L^2}.$
In addition, there is a conserved $\mathrm{U}(1)$ current with $\ka = 1$ and $J_{\pm}(x) = \! \fermionWick{ \psi^+_{\pm}(x) \psi^-_{\pm}(x) }$.
The primary fields consist of the fermionic fields, for which
$\D^{\pm}_{\psi^+_r} = \D^{\pm}_{\psi^-_r} = \d_{r,\pm}/2$ and
$\tau^{\pm}_{\psi^+_r} = -\tau^{\pm}_{\psi^-_r} = \d_{r,\pm}$.
\end{example}
\begin{example}[Free massless bosons on the circle]
\label{Example:Free_bosons}
Let $\rho_{r}(x) = \rho_{r}(x)^\dagger$ for $r = \pm$ be right- and left-moving bosonic fields satisfying the commutation relations
\begin{equation}
\label{rho_rho_comm_rels}
\bigl[ \rho_{r}(x), \rho_{r'}(x') \bigr]
= \frac{r}{2\pi\ii} \d_{r, r'} \d'(x-x')
\end{equation}
and periodic boundary conditions $\rho_{r}(x + L) = \rho_{r}(x)$.
This also defines a CFT with $c = 1$ given by the Hamiltonian
\begin{equation}
\label{H_free_bosons}
H_{B}
= \sum_{r=\pm} \int_{-L/2}^{L/2} \dd x\, \pi \vB \! \fermionWick{ \rho_{\pm}(x)^2 }
	- \frac{\pi\vB}{6L},
\end{equation}
where $\vB > 0$ and (by abuse of notation) $\fermionWick{\cdots}$ indicates (boson) Wick ordering.
Here,
$T_{\pm}(x)
= \pi \! \fermionWick{ \rho_{\pm}(x)^2 } - {\pi}/{12L^2}$, and there is a conserved $\mathrm{U}(1)$ current with $\ka = 1$ and $J_{\pm}(x) = \rho_{\pm}(x)$.
The primary fields consist of vertex operators involving the latter, see, e.g., \cite{FMS}.
\end{example}
\begin{remark}
\label{Remark:Bosonization}
There is a well-known equivalence between the models in Examples~\ref{Example:Free_fermions} and~\ref{Example:Free_bosons}, commonly referred to as bosonization.
The densities
$\rho_{\pm}(x)
= \! \fermionWick{ \psi^+_{\pm}(x) \psi^-_{\pm}(x) } \!$
can be shown to satisfy the bosonic properties in \eqref{rho_rho_comm_rels}, and, setting $\vB = \vF$, one can establish an operator identity between the Hamiltonians in \eqref{H_free_fermions} and \eqref{H_free_bosons} known as Kronig's identity.
This can be used to express the fermionic model in Example~\ref{Example:Free_fermions} as the bosonic one in Example~\ref{Example:Free_bosons}.
Similarly, there is an operator identity relating the fermionic fields to vertex operators involving the bosons. 
For details and precise statements, see, e.g., \cite{LaMo1} and references therein.
\end{remark}
\begin{example}[Local Luttinger model]
\label{Example:Luttinger_model}
This is a CFT with $c = 1$ describing interacting massless fermions formally given by the Hamiltonian \cite{Voit, SCP}
\begin{equation}
\label{H_local_Luttinger_model}
H
= H_{F}
	+ \sum_{r, r' = \pm} \int_{-L/2}^{L/2} \dd x\,
		\Bigl[ \d_{r, -r'} \frac{g_{2} \pi\vF}{2} + \d_{r, r'} \frac{g_{4} \pi\vF}{2} \Bigr]
		\! \fermionWick{ \psi^+_{r\ppr}(x) \psi^-_{r\ppr}(x) } \!
		\! \fermionWick{ \psi^+_{r'}(x) \psi^-_{r'}(x) } \!
	- L \cE_0
\end{equation}
with $H_{F}$ in \eqref{H_free_fermions} for $\psi^\pm_{r}(x)$ satisfying \eqref{psi_psi_anticomm_rels} and dimensionless coupling constants $g_{2}$ and $g_{4}$ satisfying $|g_{2}| < 2 + g_{4}$.
In the above, $\cE_0$ is a (diverging) constant subtracting the ground-state energy density up to the contribution $-\pi(v-\vF)/6L^2$, where $v = \vF \sqrt{ (1 + g_{4}/2)^2 - (g_{2}/2)^2 }$.
The model also possesses a conserved $\mathrm{U}(1)$ current with $\ka = K = \sqrt{ ({2 + g_{4} - g_{2}})/({2 + g_{4} + g_{2}}) }$.
In the local limit, ultraviolet divergencies are generated, which require additive and multiplicative renormalizations of the Hamiltonian [the term $-L\cE_0$ in \eqref{H_local_Luttinger_model}] and the fermionic fields, respectively.
In bosonized form (see Remark~\ref{Remark:Bosonization}), the Hamiltonian is $H = \int_{-L/2}^{L/2} \dd x\, v \bigl[ T_{+}(x) + T_{-}(x) \bigr]$ with
$T_{\pm}(x)
= (\pi/K) \! \fermionWick{ J_{\pm}(x)^2 } - {\pi}/{12L^2}$
and
$J_{\pm}(x)
= (1 + K) \rho_{\pm}(x)/2 + (1 - K) \rho_{\mp}(x)/2$,
where $\fermionWick{\cdots}$ indicates Wick ordering with respect to the interacting ground state.
The renormalized fermionic fields are primary with $\D^{\pm}_{\psi^+_r} = \D^{\pm}_{\psi^-_r} = (1 \pm rK)^2/8K$ and
$\tau^{\pm}_{\psi^+_r} = -\tau^{\pm}_{\psi^-_r} = (1 \pm rK)/2$.
\end{example}

We recall that the CFT in Example~\ref{Example:Luttinger_model} gives an effective description of the quantum $XXZ$ spin chain (close to but not exactly at half filling) with the so-called Luttinger parameter $K$ corresponding to the anisotropy $\D$.%
\footnote{%
\label{Footnote:K_and_Delta}%
If $\kF$ denotes the Fermi momentum and $a$ denotes the lattice spacing, then $K = {1}/{\sqrt{ 1 + {4\D}\sin(a\kF)/{\pi} }}$ for $a\kF$ close but not exactly equal to $\pi/2$, since at half filling there is an umklapp interaction term (with coupling constant proportional to $\Delta$) in the effective description that spoils exact solvability.%
}
Similarly, one application of the CFT in Example~\ref{Example:Free_fermions} is as the effective description of the quantum $XX$ spin chain ($\D = 0$).

Other CFTs to which our considerations apply include so-called minimal models and $k$-level Wess-Zumino-Witten models, cf., e.g., \cite{FMS, KnZa}.


\section{Non-equilibrium recipe}
\label{Sec:Non-equilibrium_recipe}


We recall that by inhomogeneous CFT we mean a unitary 1+1-dimensional CFT with Hamiltonian
\begin{equation}
\label{H_iCFT}
H
= \int_{-L/2}^{L/2} \dd x\, \cE(x),
\quad
\cE(x)
= v(x) \bigl[ T_{+}(x) + T_{-}(x) \bigr],
\end{equation}
where $v(x) = v(x + L) > 0$ is a smooth function and $T_{\pm}(x)$ satisfy \eqref{Virasoro_alg_pos_space}.
As usual, $H$ is the charge associated with energy conservation and $\cE(x)$ denotes the energy density.
Supposing (as we do) that our CFT has a conserved $\mathrm{U}(1)$ current, the associated total conserved charge is
\begin{equation}
\label{Q_iCFT}
Q
= \int_{-L/2}^{L/2} \dd x\, \rho(x),
\quad
\rho(x)
= J_{+}(x) + J_{-}(x),
\end{equation}
where $J_{\pm}(x)$ satisfy \eqref{Virasoro_and_u1_current_alg_pos_space} and $\rho(x)$ denotes the particle density.

As explained in Sect.~\ref{Sec:Introduction}, given smooth functions $\b(x) = \b(x + L) > 0$ and $\m(x) = \m(x + L)$, we are interested in expectations of the form
\begin{equation}
\label{G_iCFT}
\langle \cdots \rangle_{\mathrm{neq}}
= \frac{
		\Tr \bigl[
			\ee^{-G} (\cdots)
		\bigr]
	}{
		\Tr \bigl[ \ee^{-G} \bigr]
	},
\quad
G
= \int_{-L/2}^{L/2} \dd x\, \b(x)
	\bigl[ \cE(x) - \m(x) \rho(x) \bigr]
\end{equation}
for operators evolving under the dynamics given by $H$ in \eqref{H_iCFT}.
For later reference, define
\begin{subequations}
\label{fgh}
\begin{align}
f(x)
& = \int_{0}^{x} \dd x'\, \frac{v_{0}}{v(x')},
& \frac{1}{v_{0}}
& = \frac{1}{L} \int_{-L/2}^{L/2} \dd x'\, \frac{1}{v(x')},
		\label{f_v_x} \\
g(x)
& = \int_{0}^{x} \dd x'\, \frac{v_{0}\b_{0}}{v(x')\b(x')},
& \frac{1}{v_{0}\b_{0}}
& = \frac{1}{L} \int_{-L/2}^{L/2} \dd x'\, \frac{1}{v(x')\b(x')},
		\label{g_vbeta_x} \\
h(x)
& = \int_{0}^{x} \dd x'\, \frac{\m(x')\b(x') - \m_{0}\b_{0}}{v(x')\b(x')},
& \frac{\m_{0}}{v_{0}}
& = \frac{1}{L} \int_{-L/2}^{L/2} \dd x'\, \frac{\m(x')}{v(x')}.
		\label{h_vmu_x}
\end{align}
\end{subequations}
In Sect.~\ref{Sec:Main_tools:Proof_of_result}, we prove the following recipe for computing all such non-equilibrium expectations:

\begin{proposition}
\label{Proposition:Main}
Define $H$, $Q$, and $G$ as in \eqref{H_iCFT}\textnormal{--}\eqref{G_iCFT}
as well as $f(x)$, $g(x)$, $h(x)$, $v_{0}$, $\b_{0}$, and $\m_{0}$ as in \eqref{fgh}.
Let\,%
\footnote{%
This is to avoid any confusion of notation, reserving $\cO(\cdot, \cdot)$ to the dependence on the right- and left-moving coordinates, cf.\ the primary fields in Sects.~\ref{Sec:Prerequisites:Conformal_transformations} and~\ref{Sec:Prerequisites:Gauge_transformations}.%
}
\begin{equation}
\label{H_dynamics}
\cO(x; t) = \ee^{\ii H t} \cO(x) \ee^{-\ii H t}
\end{equation}
denote the inhomogeneous time evolution for any local operator $\cO(x)$ and let
\begin{equation}
\label{H_sCFT}
\langle \cdots \rangle_{0}
= \frac{
		\Tr \bigl[
			\ee^{- \b_{0} (H_{0} - \m_{0} Q_{0})} (\cdots)
		\bigr]
	}{
		\Tr \bigl[ \ee^{-\b_{0} (H_{0} - \m_{0} Q_{0})} \bigr]
	},
\quad
H_{0}
= \int_{-L/2}^{L/2} \dd x\, v_{0} \bigl[ T_{+}(x) + T_{-}(x) \bigr],
\quad
Q_{0} = Q
\end{equation}
denote the translation-invariant expectation corresponding to $\langle \cdots \rangle_{\mathrm{neq}}$ in \eqref{G_iCFT}.
Moreover, define
\begin{equation}
\label{x_pm_tilde}
\tilde{x}^{\pm}
= f^{-1}(f(x) \pm v_{0}t)
\end{equation}
and introduce the following:
\begin{itemize}

\item
For the components of the energy-momentum tensor,
\begin{align}
\tilde{T}_{\pm}(x; t)
& = 	\biggl( \frac{v_{0}\b_{0}}{v(x)\b(\tilde{x}^{\mp})} \biggr)^2
			T_{\pm}(g(\tilde{x}^{\mp}))
		+ \frac{v_{0}\b_{0}
					[\m(\tilde{x}^{\mp})\b(\tilde{x}^{\mp}) - \m_{0}\b_{0}]}
				{[v(x)\b(\tilde{x}^{\mp})]^2}
			J_{\pm}(g(\tilde{x}^{\mp}))
			\nonumber \\
& \quad
		+ \frac{\ka}{4\pi}
			\biggl(
				\frac{\m(\tilde{x}^{\mp})\b(\tilde{x}^{\mp})
				- \m_{0}\b_{0}}{v(x)\b(\tilde{x}^{\mp})}
			\biggr)^2
		- \frac{\cT(\tilde{x}^{\mp}) + \cS(x)}{2 v(x)^2}
		\label{Tpm_tilde}
\end{align}
with
\begin{subequations}
\label{cS_cT}
\begin{align}
\cS(x)
& = - \frac{c v(x)^2}{12\pi}
			\Biggl[
				\frac{v''(x)}{v(x)}
				- \frac{1}{2} \biggl( \frac{v'(x)}{v(x)} \biggr)^2
			\Biggr],
			\label{cS} \\
\cT(x)
& = - \frac{c v(x)^2}{12\pi}
			\Biggl[
				\frac{\b''(x)}{\b(x)}
				- \frac{1}{2} \biggl( \frac{\b'(x)}{\b(x)} \biggr)^2
	 			+ \frac{v'(x)}{v(x)} \frac{\b'(x)}{\b(x)}
			\Biggr].
			\label{cT}
\end{align}
\end{subequations}

\item
For the components of the $\mathrm{U}(1)$ current,
\begin{equation}
\label{Jpm_tilde}
\tilde{J}_{\pm}(x; t)
= \frac{v_{0}\b_{0}}{v(x)\b(\tilde{x}^{\mp})} J_{\pm}(g(\tilde{x}^{\mp}))
	+ \frac{\ka}{2\pi}
		\frac{\m(\tilde{x}^{\mp})\b(\tilde{x}^{\mp}) - \m_{0}\b_{0}}{v(x)\b(\tilde{x}^{\mp})}.
\end{equation}

\item
For any primary field $\Phi$ \textnormal{[}satisfying \eqref{U_Phi_Ui_prerequisites} and \eqref{V_Phi_Vi_prerequisites}\textnormal{]} with weights
$(\D^+_{\Phi}, \D^-_{\Phi})$ and $(\tau^+_{\Phi}, \tau^-_{\Phi})$,
\begin{equation}
\label{Phi_tilde}
\tilde{\Phi}(x; t)
= \ee^{-\ii [h(\tilde{x}^-) \tau^+_{\Phi} - h(\tilde{x}^+) \tau^-_{\Phi}]}
	\biggl( \frac{v_{0}\b_{0}}{v(x)\b(\tilde{x}^-)} \biggr)^{\D^+_{\Phi}}
	\biggl( \frac{v_{0}\b_{0}}{v(x)\b(\tilde{x}^+)} \biggr)^{\D^-_{\Phi}}
	\Phi(g(\tilde{x}^-), g(\tilde{x}^+)).
\end{equation}

\end{itemize}
Then
\begin{equation}
\label{recipe}
\langle \cO_{1}(t_{1}) \ldots \cO_{n}(t_{n}) \rangle_{\mathrm{neq}}
= \langle \tilde{\cO}_{1}(t_{1}) \ldots \tilde{\cO}_{n}(t_{n}) \rangle_{0}
\end{equation}
with $O_{j}(t_{j}) = \ee^{\ii H t_{j}} \cO_{j} \ee^{-\ii H t_{j}}$ and $\tilde{O}_{j}(t_{j})$ given by \eqref{Tpm_tilde}\textnormal{--}\eqref{Phi_tilde} for all $\cO_{j}$ in the algebra of operator-valued distributions generated by the components $T_{\pm}$ and $J_{\pm}$ together with all primary fields with products restricted to non-coincident points in space.
\end{proposition}

For clarity, one can write \eqref{x_pm_tilde} as $\tilde{x}^{\pm} = \tilde{x}_{t}^{\pm}(x)$ with $\tilde{x}_{t}^{\pm}(x) = f^{-1}(f(x) \pm v_{0}t)$.
Proposition~\ref{Proposition:Main} makes manifest that the time evolution is entirely encoded in $\tilde{x}_{t}^{\pm}(x)$.
Moreover, it is straightforward to show that the latter satisfy the group property
$\tilde{x}_{t_{1}+t_{2}}^{\pm}(x)
= \tilde{x}_{t_{1}}^{\pm}(\tilde{x}_{t_{2}}^{\pm}(x))
= \tilde{x}_{t_{2}}^{\pm}(\tilde{x}_{t_{1}}^{\pm}(x))$
for $t_{1}, t_{2} \in \mathbb{R}$
and the equations of motion
\begin{equation}
\label{dxpm_dx}
\partial_{t} \tilde{x}_{t}^{\pm}(x)
= \pm v(x) \partial_{x} \tilde{x}_{t}^{\pm}(x),
\qquad
\tilde{x}_{0}^{\pm}(x) = x.
\end{equation}
This and the above justify defining $\tilde{x}^{\pm} = \tilde{x}_{t}^{\pm}(x)$ as coordinates generalizing the usual light-cone coordinates $x^{\pm}$ to the inhomogeneous dynamics given by $H$ in \eqref{H_iCFT}.

\begin{remark}
By extending the algebra of operator-valued distributions in Proposition~\ref{Proposition:Main} to coincident points and properly normal ordering the products, the statement can in principle be generalized to include also all descendent fields corresponding to each primary field.
\end{remark}
\begin{remark}
\label{Remark:gSSD_CFT}
As an example, consider the ``regularized'' SSD CFT given by $v(x)$ in \eqref{v_x_gSSD_CFT}.
Then
\begin{equation}
v_{0} = v \sqrt{1 - g^2},
\quad
f(x)
= \frac{L}{\pi}
	\arctan \biggl[
		\sqrt{\frac{1-g}{1+g}}
		\tan \! \left( \frac{\pi x}{L} \right)
	\biggr]
\end{equation}
and
\begin{equation}
\tilde{x}^{\pm}
= \frac{L}{\pi}
	\arctan
	\biggl[
		\sqrt{\frac{1+g}{1-g}}
		\tan
		\biggl(
			\arctan
			\biggl[
				\sqrt{\frac{1-g}{1+g}}
				\tan \! \left( \frac{\pi x}{L} \right)
			\biggr]
			\pm \sqrt{1 - g^2} \frac{\pi vt}{L}
		\biggr)
	\biggr].
\end{equation}
Since $v_{0} \to 0$ as $g \to 1^{-}$, $f(x)$ is not well defined in this limit, thus not an element of $\wDiff_{+}(S^1)$, cf.\ Sect.~\ref{Sec:Main_tools}.
However, the limit can be taken for $\tilde{x}^{\pm}$,
\begin{equation}
\lim_{g \to 1^{-}} \tilde{x}^{\pm}
= \frac{L}{\pi}
	\arctan
	\biggl[
		\tan \! \left( \frac{\pi x}{L} \right)
		\pm \frac{2\pi vt}{L}
	\biggr],
\end{equation}
implying that the points $x = \pm L/2$ do not evolve for SSD CFT, consistent with $v(\pm L/2) = (1-g)v$.
Lastly, for $|x|$ finite, $\tilde{x}^{\pm} \to x \pm (1+g) vt$ as $L \to \io$, i.e., this example then only amounts to a velocity renormalization, consistent with $\lim_{L \to \io} v(x) = (1+g) v$ away from $x = \pm L/2$.
\end{remark}


\section{Applications}
\label{Sec:Applications}


In this section, we present a number of exact analytical results obtained using Proposition~\ref{Proposition:Main}.
For later reference, a superscript $^\io$ will be used to indicate expectations in the limit $L \to \io$, e.g.,
$\langle \cO(x; t) \rangle_{\mathrm{neq}}^{\io}
= \lim_{L\to\io} \langle \cO(x; t) \rangle_{\mathrm{neq}}$,
and a superscript $^{c}$ to denote the connected part, e.g.,
$\langle \cO_{1}(x_{1}; t_{1}) \cO_{2}(x_{2}; t_{2}) \rangle_{\mathrm{neq}}^{c}
= \langle \cO_{1}(x_{1}; t_{1})\cO_{2}(x_{2}; t_{2}) \rangle_{\mathrm{neq}}
	- \langle \cO_{1}(x_{1}; t_{1}) \rangle_{\mathrm{neq}}
		\langle \cO_{2}(x_{2}; t_{2}) \rangle_{\mathrm{neq}}$.

\begin{remark}
\label{Remark:v_0_b_0_m_0}
It is important to note that $v_{0}$, $\b_{0}$, and $\m_{0}$ in \eqref{fgh} are defined in the finite volume and will not directly appear in any infinite-volume results:
They can be shown to cancel in every such computation, leaving only the dependence on $v(x)$, $\beta(x)$, and $\mu(x)$ at or between spacetime points where local operators are evaluated, as intuitively expected and mentioned already in Sect.~\ref{Sec:Introduction}.
Alternatively, if needed, the finite-volume constants can be replaced by new $v > 0$, $\b > 0$, and $\m \in \mathbb{R}$ in the infinite volume that can differ from the former up to $O(L^{-1})$ contributions, as discussed in \cite{GLM, GaKo}.
Moreover, in what follows, the generalized light-cone coordinates $\tilde{x}^{\pm}$ in \eqref{x_pm_tilde} can be assumed given by an $f(x)$ with $v_{0}$ replaced by such a $v$ in the infinite volume: In fact, the value of $v$ does not matter for $\tilde{x}^{\pm} = \tilde{x}^{\pm}_{t}(x)$ in the limit $L \to \io$ since these coordinates can then simply be viewed as defined by \eqref{dxpm_dx} and thus only depend directly on $v(x)$.
\end{remark}


\subsection{Densities and currents}
\label{Sec:Applications:Densities_and_currents}


The energy density operator $\cE(x)$ is given in \eqref{H_iCFT}, and the corresponding heat current operator
\begin{equation}
\label{cJ_iCFT}
\cJ(x) = v(x)^2 \bigl[ T_{+}(x) - T_{-}(x) \bigr]
\end{equation}
can be identified from that the pair must satisfy a continuity equation.
(This determines both up to trivial c-number contributions.)
It follows using \eqref{dxpm_dx} [and \eqref{HtTpmHt}] that
\begin{subequations}
\label{iCFT_cEcJ_cont_eq}
\begin{gather}
\partial_t \cE(x) + \partial_x \cJ(x)
= 0,
	\label{iCFT_cEcJ_cont_eq_1} \\
\partial_t \cJ(x) + v(x) \partial_x \bigl[ v(x) \cE(x) + \cS(x) \bigr]
= 0
	\label{iCFT_cEcJ_cont_eq_2}
\end{gather}
\end{subequations}
with $\cS(x)$ in \eqref{cS}.
Here, \eqref{iCFT_cEcJ_cont_eq_1} is the usual continuity equation associated with energy conservation.
However, \eqref{iCFT_cEcJ_cont_eq_2} is not the same as in standard CFT, see Remark~\ref{Remark:New_momentum_and_axial_charge_conservation}.
In particular, the latter implies that the total current
$\int_{-L/2}^{L/2} \dd x\, \cJ(x; t)$
is not conserved if $v(x) \neq v$.

Again, supposing (as we do) that our CFT has a conserved $\mathrm{U}(1)$ current, the particle density operator $\rho(x)$ is given in \eqref{Q_iCFT}, and the corresponding charge current operator
\begin{equation}
\label{j_iCFT}
j(x) = v(x) \bigl[ J_{+}(x) - J_{-}(x) \bigr]
\end{equation}
can be identified by the same argument as above: Using \eqref{dxpm_dx} [and \eqref{HtJpmHt}], we obtain%
\begin{subequations}
\label{iCFT_rj_cont_eq}
\begin{gather}
\partial_t \rho(x) + \partial_x j(x)
= 0,
	\label{iCFT_rj_cont_eq_1} \\
\partial_t j(x) + v(x) \partial_x \bigl[ v(x)\rho(x) \bigr]
= 0.
	\label{iCFT_rj_cont_eq_2}
\end{gather}
\end{subequations}
Similar to above, \eqref{iCFT_rj_cont_eq_1} is the continuity equation associated with particle number conservation, while \eqref{iCFT_rj_cont_eq_2} implies that the total charge current is not conserved, see Remark~\ref{Remark:New_momentum_and_axial_charge_conservation}.

In Appendix~\ref{App:Computational_details:cEcJ_rj_iCFT_neq}, we derive the following results:
\begin{subequations}
\label{cEcJ_rj_iCFT_neq}
\begin{align}
\langle \cE(x; t) \rangle_{\mathrm{neq}}^{\io}
& = \frac{F(\tilde{x}^-) + F(\tilde{x}^+)}{2v(x)} - \frac{\cS(x)}{v(x)},
& \langle \cJ(x; t) \rangle_{\mathrm{neq}}^{\io}
& = \frac{F(\tilde{x}^-) - F(\tilde{x}^+)}{2},
		\label{cEcJ_iCFT_neq} \\
\langle \rho(x; t) \rangle_{\mathrm{neq}}^{\io}
& = \frac{G(\tilde{x}^-) + G(\tilde{x}^+)}{2v(x)},
& \langle j(x; t) \rangle_{\mathrm{neq}}^{\io}
& = \frac{G(\tilde{x}^-) - G(\tilde{x}^+)}{2},
		\label{rj_iCFT_neq}
\end{align}
\end{subequations}
where $\tilde{x}^{\pm} = \tilde{x}^{\pm}_{t}(x)$ are given by \eqref{dxpm_dx} (cf.\ Remark~\ref{Remark:v_0_b_0_m_0}) and
\begin{equation}
\label{F_x_G_x_iCFT_neq}
F(x)
= \frac{\pi c}{6 \b(x)^2} + \frac{\ka \m(x)^2}{2\pi} - \cT(x),
\quad
G(x)
= \frac{\ka \m(x)}{\pi}
\end{equation}
with $\cS(x)$ and $\cT(x)$ in \eqref{cS_cT}.
Note that these results generalize the ones in \cite{GLM} to inhomogeneous dynamics.%
\footnote{%
Our conventions for $F(x)$ and $G(x)$ differ from those in \cite{GLM} by a velocity factor.%
}

\begin{remark}[Interpretations of \eqref{iCFT_cEcJ_cont_eq_2} and \eqref{iCFT_rj_cont_eq_2}]
\label{Remark:New_momentum_and_axial_charge_conservation}
As usual, the Hamiltonian $H$ in \eqref{H_iCFT} is the generator of time translations and the (usual) total momentum operator
$P = \int_{-L/2}^{L/2} \dd x\, [T_{+}(x) - T_{-}(x)]$ is the generator of spatial ones.
These do not commute for non-constant $v(x)$.
Indeed, it is straightforward to show that
$\partial_{t} P
= \ii [H, P]
= - \int_{-L/2}^{L/2} \dd x\, [{v'(x)}/{v(x)}] \cE(x)$.
This implies that (usual) momentum is not conserved in inhomogeneous CFT.
However, one can define a new operator
$\int_{-L/2}^{L/2} \dd x\, [v(x)/v_{0}] [T_{+}(x) - T_{-}(x)]$
generating position-dependent spatial translations given by $x \to x' = x + \e v(x)/v_{0} + o(\e)$ for infinitesimally small $\e$.
This operator is conserved with the corresponding continuity equation given by \eqref{iCFT_cEcJ_cont_eq_2}.
Similarly, \eqref{iCFT_rj_cont_eq_2} corresponds to the continuity equation associated with the axial current with $\rho_{A}(x) = [J_{+}(x) - J_{-}(x)]/\ka$ and $j_{A}(x) = v(x) [J_{+}(x) + J_{-}(x)]/\ka$.
\end{remark}
%


\subsection{Correlation functions}


In Appendix~\ref{App:Computational_details:cJ_j_combos_iCFT_neq_and_psipsi_iLLM_neq}, we derive the following results for the connected current-current correlation functions in the thermodynamic limit:
\begin{subequations}
\label{cJ_j_combos_iCFT_neq}
\begin{align}
\langle \cJ(x_{1}; t_{1}) \cJ(x_{2}; t_{2}) \rangle_{\mathrm{neq}}^{c, \io}
& = \sum_{r=\pm}
		\Biggl[
			\frac{\pi^2 c}
				{8\b(\tilde{x}^{-r}_{1})^2\b(\tilde{x}^{-r}_{2})^2
					\sinh^4
					\bigl(
						\pi \int_{\tilde{x}^{-r}_{2}}^{\tilde{x}^{-r}_{1}}
						\dd x'\, v(x')^{-1} \b(x')^{-1}
					\bigr)} \nonumber \\
& \qquad\quad +
			\frac{-\ka\m(\tilde{x}^{-r}_{1})\m(\tilde{x}^{-r}_{2})}
				{4\b(\tilde{x}^{-r}_{1})\b(\tilde{x}^{-r}_{2})
					\sinh^2
					\bigl(
						\pi \int_{\tilde{x}^{-r}_{2}}^{\tilde{x}^{-r}_{1}}
						\dd x'\, v(x')^{-1} \b(x')^{-1}
					\bigr)}
		\Biggr],
		\label{cJcJ_iCFT_neq} \\
\langle j(x_{1}; t_{1}) j(x_{2}; t_{2}) \rangle_{\mathrm{neq}}^{c, \io}
& = \sum_{r=\pm}
		\frac{-\ka}
			{4\b(\tilde{x}^{-r}_{1})\b(\tilde{x}^{-r}_{2})
				\sinh^2
				\bigl(
					\pi \int_{\tilde{x}^{-r}_{2}}^{\tilde{x}^{-r}_{1}}
					\dd x'\, v(x')^{-1} \b(x')^{-1}
				\bigr)},
		\label{jj_iCFT_neq} \\
\langle \cJ(x_{1}; t_{1}) j(x_{2}; t_{2}) \rangle_{\mathrm{neq}}^{c, \io}
& = \sum_{r=\pm}
		\frac{-\ka\m(\tilde{x}^{-r}_{1})}
			{4\b(\tilde{x}^{-r}_{1})\b(\tilde{x}^{-r}_{2})
				\sinh^2
				\bigl(
					\pi \int_{\tilde{x}^{-r}_{2}}^{\tilde{x}^{-r}_{1}}
					\dd x'\, v(x')^{-1} \b(x')^{-1}
				\bigr)},
		\label{cJj_iCFT_neq} \\
\langle j(x_{1}; t_{1}) \cJ(x_{2}; t_{2}) \rangle_{\mathrm{neq}}^{c, \io}
& = \sum_{r=\pm}
		\frac{-\ka\m(\tilde{x}^{-r}_{2})}
			{4\b(\tilde{x}^{-r}_{1})\b(\tilde{x}^{-r}_{2})
				\sinh^2
				\bigl(
					\pi \int_{\tilde{x}^{-r}_{2}}^{\tilde{x}^{-r}_{1}}
					\dd x'\, v(x')^{-1} \b(x')^{-1}
				\bigr)},
		\label{jcJ_iCFT_neq}
\end{align}
\end{subequations}
where $\tilde{x}^{\pm}_{j} = \tilde{x}^{\pm}_{t_{j}}(x_{j})$ are given by \eqref{dxpm_dx} (cf.\ Remark~\ref{Remark:v_0_b_0_m_0}).
These infinite-volume results hold for any inhomogeneous CFT; the last three assuming there is a conserved $\mathrm{U}(1)$ current.
Note that all $n$-point current correlation functions can in principle be computed, if known in the homogeneous case.

Given primary fields $\Phi_{j}$ ($j = 1, \ldots, n$) [satisfying \eqref{U_Phi_Ui_prerequisites} and \eqref{V_Phi_Vi_prerequisites}] with weights
$(\D^{+}_{\Phi_{j}}, \D^{-}_{\Phi_{j}})$ and $(\tau^{+}_{\Phi_{j}}, \tau^{-}_{\Phi_{j}})$,
it follows from Proposition~\ref{Proposition:Main} that their correlation functions are given by
\begin{align}
\biggl\langle \prod_{j = 1}^{n} \Phi_{j}(x_{j}; t_{j}) \biggr\rangle_{\mathrm{neq}}
& = \exp \biggl(
			- \ii \sum_{j = 1}^{n}
				\bigl[
					h(\tilde{x}^-_{j}) \tau^+_{\Phi_{j}} - h(\tilde{x}^+_{j}) \tau^-_{\Phi_{j}}
				\bigr]
		\biggr) \nonumber \\
& \quad \times
		\biggl[
			\prod_{j = 1}^{n}
			\biggl( \frac{v_{0}\b_{0}}{v(x_{j})\b(\tilde{x}^-_{j})} \biggr)^{\D^+_{\Phi_{j}}}
			\biggl( \frac{v_{0}\b_{0}}{v(x_{j})\b(\tilde{x}^+_{j})} \biggr)^{\D^-_{\Phi_{j}}}
		\biggr]
		\biggl\langle
			\prod_{j = 1}^{n}
			\Phi_{j}(g(\tilde{x}^-_{j}), g(\tilde{x}^+_{j}))
		\biggr\rangle_{0}.
\end{align}

Important applications of correlation functions of primary fields include to study the propagation of excitations \cite{CaCa1} or compute the entanglement entropy \cite{CaCa2, CCaD}.
For instance, studying two-point correlation functions of primary fields, one finds that excitations propagate along curved light-cones given by $v(x)$ following a quantum quench \cite{DSC}, while for Floquet systems described by inhomogeneous CFT, excitations can be shown to accumulate at unstable fixed points of the corresponding trajectories if the system is heating, see, e.g., \cite{LCTTNC1, LapMoo}.

As examples of primary fields, we consider the fermionic fields $\psi_r^{\pm}$ in the inhomogeneous local Luttinger model, cf.\ Example~\ref{Example:Luttinger_model}.
Supposing that these are properly renormalized, we show in Appendix~\ref{App:Computational_details:cJ_j_combos_iCFT_neq_and_psipsi_iLLM_neq} that the explicit fermion two-point correlation functions are
\begin{align}
\langle \psi^+_{r}(x_{1}; t_{1}) \psi^-_{r'}(x_{2}; t_{2}) \rangle_{\mathrm{neq}}^{\io}
& = \d_{r, r'}
		\frac{1}{2\pi} \,
		\exp \biggl(
			  \ii \tau^+_{\psi^-_r} \int_{\tilde{x}^-_{2}}^{\tilde{x}^-_{1}} \dd x'\, \frac{\m(x')}{v(x')}
			- \ii \tau^-_{\psi^-_r} \int_{\tilde{x}^+_{2}}^{\tilde{x}^+_{1}} \dd x'\, \frac{\m(x')}{v(x')}
		\biggr)
		\nonumber \\
& \quad \times
		\left(
			\frac{
				\ii \pi
			}{
				\sqrt{v(x_{1})v(x_{2})\b(\tilde{x}^-_{1})\b(\tilde{x}^-_{2})}
				\sinh
				\bigl(
					\pi \int_{\tilde{x}^{-}_{2}}^{\tilde{x}^{-}_{1}}
					\dd x'\, v(x')^{-1} \b(x')^{-1}
				\bigr)
			}
		\right)^{2\D^+_{\psi^-_r}}
		\nonumber \\
& \quad \times
		\left( 
			\frac{
				-\ii \pi
			}{
				\sqrt{v(x_{1})v(x_{2})\b(\tilde{x}^+_{1})\b(\tilde{x}^+_{2})}
				\sinh
				\bigl(
					\pi \int_{\tilde{x}^{+}_{2}}^{\tilde{x}^{+}_{1}}
					\dd x'\, v(x')^{-1} \b(x')^{-1}
				\bigr)
			}
	\right)^{2\D^-_{\psi^-_r}},
	\label{psipsi_iLLM_neq}
\end{align}
where $\tilde{x}^{\pm}_{j} = \tilde{x}^{\pm}_{t_{j}}(x_{j})$ are given by \eqref{dxpm_dx} (cf.\ Remark~\ref{Remark:v_0_b_0_m_0})
and $(\D^{+}_{\psi^{\pm}_r}, \D^{-}_{\psi^{\pm}_r})$ and
$(\tau^{+}_{\psi^{\pm}_r}, \tau^{-}_{\psi^{\pm}_r})$ are given in Example~\ref{Example:Luttinger_model}.


\subsection{Conductivities}
\label{Sec:Applications:Conductivities}


To compute conductivities in a unified way, we rearrange the thermodynamic variables into $\m_{1} = \b\m$ and $\m_{2} = -\b$ and label the densities and currents as $\rho_{1} = \rho$, $\rho_{2} = \cE$, $j_{1} = j$, and $j_{2} = \cJ$, respectively.%
\footnote{%
\label{Footnote:Higher_conserved_charges}%
The order by which we label quantities for electrical and heat transport is different from that in \cite{GLM}.
Our list can also be generalized to higher conserved charges if such exist, i.e., to $Q_{n} = \int \dd x\, \rho_{n}$ with density $\rho_{n}$ and current $j_{n}$ satisfying $\partial_t \rho_{n} + \partial_x j_{n} = 0$ for $n \geq 3$.%
}
Let $\ka_{mn}(\o)$ denote conductivities as functions of frequency $\o$.
Following \cite{Kubo}, we define them as linear-response functions measuring the change in the total current $\int \dd x\, j_{m}$ due to a unit-pulse perturbation in $\m_{n}$, which we recall is the thermodynamic conjugate to $\rho_{n}$, see Appendix~\ref{App:Linear-response_theory}.
In our case, $\m_{n}(x) = \m_{n} + \d\m_{n} W(x)$ with the spatial dependence of the perturbations $\d\m_{n} W(x)$ given by a function $W(x)$ describing an overall kink-like profile such that $\lim_{x\to\mp\io} W(x) = \pm 1/2$ in the infinite volume.%
\footnote{%
See \cite{GLM} (and also Appendix~\ref{App:Linear-response_theory:Conductivity_matrix}) for a discussion of how this is compatible with our periodic boundary conditions in the finite volume.%
}
As in \eqref{ka_th_s_el}, we recall that, on general grounds,
\begin{equation}
\label{ka_mn_o}
\Re \ka_{mn}(\o)
= D_{mn} \pi \d(\o) + \Re \ka_{mn}^{\mathrm{reg}}(\o)
\end{equation}
with Drude weights $D_{mn}$ and real regular parts $\Re \ka_{mn}^{\mathrm{reg}}(\o)$.
These are important quantities characterizing the transport properties:
A non-zero $D_{mn}$ corresponds to a non-zero ballistic contribution, while a non-zero $\Re \ka_{mn}^{\mathrm{reg}}(\o)$ for $\omega = 0$ ($\neq 0$) corresponds to a non-zero normal (anomalous) diffusive contribution.

We give two different approaches by which $D_{mn}$ and $\Re \ka_{mn}^{\mathrm{reg}}(\o)$ in \eqref{ka_mn_o} can be computed.
The first is based on direct dynamical considerations and the second on a Green-Kubo formula, see Sects.~\ref{Sec:Applications:Conductivities:Dynamical_approach} and~\ref{Sec:Applications:Conductivities:Green-Kubo_approach}, respectively.
As mentioned in Sect.~\ref{Sec:Introduction}, on general grounds, both approaches must yield the same results, see Appendix~\ref{App:Linear-response_theory}.
For inhomogeneous CFT, they are
\begin{subequations}
\label{D_mn_ka_mn_o_reg_iCFT}
\begin{equation}
\left( \begin{matrix}
	D_{11} & D_{12} \\
	D_{21} & D_{22}
\end{matrix} \right)
=
\left( \begin{matrix}
	\frac{v\ka}{\pi \b}
	& \frac{v\ka \m}{\pi \b} \\
	\frac{v\ka \m}{\pi \b}
	& \frac{\pi v c}{3 \b^3} + \frac{v \ka \m^2}{\pi \b}
\end{matrix} \right)
	\label{D_mn_iCFT}
\end{equation}
and
\begin{equation}
\left( \begin{matrix}
	\Re \ka_{11}^{\mathrm{reg}}(\o) & \Re \ka_{12}^{\mathrm{reg}}(\o) \\
	\Re \ka_{21}^{\mathrm{reg}}(\o) & \Re \ka_{22}^{\mathrm{reg}}(\o)
\end{matrix} \right)
=
\left( \begin{matrix}
	\frac{\ka}{2\pi \b}
	& \frac{\ka \m}{2\pi \b} \\
	\frac{\ka \m}{2\pi \b}
	& \frac{\pi c}{6 \b^3} \left[ 1 + \left( \frac{\o\b}{2\pi} \right)^2 \right]
		+ \frac{\ka \m^2}{2\pi \b}
\end{matrix} \right) I(\o)
	\label{ka_mn_o_reg_iCFT}
\end{equation}
\end{subequations}
with
\begin{equation}
\label{I_omega}
I(\o)
= \int_{-\io}^{\io} \dd x\!  \int_{-\io}^{\io} \dd x'\,
	\biggl( 1 - \frac{v}{v(x)} \biggr)
	\partial_{x'} \bigl[ - W(x') \bigr]
	\cos \biggl( \o \int_{x'}^{x} \frac{\dd x''}{v(x'')} \biggr),
\end{equation}	
where $v > 0$ is given by the condition that it must subtract a constant contribution from $v(x)$ in the infinite volume so that $1 - v/v(x) \in L^1(\mathbb{R})$, see Appendix~\ref{Appendix:Computational_details:I_1_I_2_results}.%
\footnote{%
This is consistent with that $v$ and $v_0$ can differ up to $O(L^{-1})$ contributions (cf.\ Remark~\ref{Remark:v_0_b_0_m_0}).%
}
This condition fixes the values of the Drude weights in \eqref{D_mn_iCFT}.


\subsubsection{Dynamical approach}
\label{Sec:Applications:Conductivities:Dynamical_approach}

Consider our initial state in \eqref{G_iCFT} defined by kink-like profiles
$\m_{1}(x) = \b(x)\m(x)$ and $\m_{2}(x) = -\b(x)$ with heights $\d\m_{1}$ and $\d\m_{2}$, respectively: $\m_{n}(x) = \m_{n} + \d\m_{n} W(x)$ with $W(x)$ above.
Then
\begin{equation}
\label{ka_mn_o_dynamic}
\ka_{mn}(\o)
= \frac{\partial}{\partial(\d\m_{n})} \int_{0}^{\io} \dd t\, \ee^{\ii \o t}
	\int_{-\io}^{\io} \dd x\, 
	\partial_t \langle j_{m}(x; t)
	\rangle_{\mathrm{neq}}^{\io}\Big|_{\d\m_{1} = \d\m_{2} = 0}.
\end{equation}
A proof is given in Appendix~\ref{App:Linear-response_theory:Conductivity_matrix}.

In Appendix~\ref{App:Computational_details:ka_mn_o_dynamic_iCFT}, we derive \eqref{D_mn_ka_mn_o_reg_iCFT} using \eqref{ka_mn_o_dynamic}, for which the only ingredients are the results for the currents in \eqref{cEcJ_rj_iCFT_neq} in the infinite volume.


\subsubsection{Green-Kubo approach}
\label{Sec:Applications:Conductivities:Green-Kubo_approach}

The conductivities can equivalently be computed using the following Green-Kubo formula:
\begin{equation}
\label{ka_mn_o_GK}
\ka_{mn}(\o)
= \frac{1}{\b} \int_{0}^{\b} \dd \t \int_{0}^{\io} \dd t\, \ee^{\ii \o t}
	\int_{-\io}^{\io} \dd x
	\int_{-\io}^{\io} \dd x'\, \partial_{x'} [-W(x')]
	\langle j_{m}(x; t) j_{n}(x'; \ii\t) \rangle_{0}^{c, \io}
\end{equation}
with $\b_{0}$ and $\m_{0}$ replaced by $\b > 0$ and $\m \in \mathbb{R}$, respectively.
A proof is given in Appendix~\ref{App:Linear-response_theory:Conductivity_matrix}.

In Appendix~\ref{App:Computational_details:ka_mn_o_GK_iCFT}, we derive \eqref{D_mn_ka_mn_o_reg_iCFT} using \eqref{ka_mn_o_GK}.
Here, besides $W(x)$ above, the only ingredients are the equilibrium current-current correlation functions in the infinite volume.
For inhomogeneous CFT, the latter are obtained from \eqref{cJ_j_combos_iCFT_neq} by setting $\b(x) = \b$ and $\m(x) = \m$.


\subsubsection{Thermal and electrical conductivities}
\label{Sec:Applications:Conductivities:Thermal_and_electrical_conductivities}

The thermal conductivity $\ka_{\mathrm{th}}(\o)$ in \eqref{ka_th} and the electrical conductivity $\s_{\mathrm{el}}(\o)$ in \eqref{s_el} are computed as the responses to changes in temperature $\b^{-1}$ and chemical potential $\m$, respectively.
Thus, it follows from \eqref{D_mn_ka_mn_o_reg_iCFT} that%
\footnote{%
By a change of variables from $\m_{1}$ and $\m_{2}$ to $\b^{-1}$ and $\m$, noting that $\partial/\partial(\b^{-1}) = - \m\b^2 \partial/\partial \m_1 + \b^2 \partial/\partial \m_2$ and $\partial/\partial \m = \b \partial/\partial \m_1$.%
}
\begin{subequations}
\begin{align}
D_{\mathrm{th}}
& = \frac{\pi v c}{3\b},
& \Re \ka_{\mathrm{th}}^{\mathrm{reg}}(\o)
& = \frac{\pi c}{6\b} \biggl[ 1 + \left( \frac{\o\b}{2\pi} \right)^2 \biggr] I(\o),
		\label{D_th_ka_th_reg_iCFT} \\
D_{\mathrm{el}}
& = \frac{v \ka}{\pi},
& \Re \s_{\mathrm{el}}^{\mathrm{reg}}(\o)
& = \frac{\ka}{2\pi} I(\o)
		\label{D_el_s_el_reg_iCFT}
\end{align}
\end{subequations}
with $I(\omega)$ in \eqref{I_omega}.
Taking their ratios implies the relations in \eqref{Wiedemann-Franz}.


\subsubsection{Alternative conductivities}
\label{Sec:Applications:Conductivities:Alternative_conductivities}

In addition to the conductivities used above, one can also define an alternative set of conductivities $\ka_{mn}(x, t; x')$ measuring the change in the current $j_{m}(x)$ at time $t$ following a unit-pulse perturbation in the profile $\m_{n}(x')$, see Appendix~\ref{App:Linear-response_theory:Alternative_linear-response_functions}.
Fourier transforming in space and time, these can, for instance, be computed using the following Green-Kubo formula:%
\footnote{%
Note that the spatial Fourier transform involves $f(x) - f(x') = \int_{x'}^{x} \dd x''\,  v/v(x'')$, cf.\ \eqref{I_omega}.
This is natural since the system is inhomogeneous.%
}
\begin{equation}
\label{ka_mn_p_o_xp_GK}
\ka_{mn}(p, \o; x')
= \frac{1}{\b} \int_{0}^{\b} \dd \t \int_{0}^{\io} \dd t\, \ee^{\ii \o t}
	\int_{-\io}^{\io} \dd x\, \ee^{-\ii p[f(x) - f(x')]}
	\langle j_{m}(x; t) j_{n}(x'; \ii\t) \rangle_{0}^{c, \io}
\end{equation}
with $\b_{0}$ and $\m_{0}$ replaced by $\b > 0$ and $\m \in \mathbb{R}$, respectively.

In analogy with \eqref{ka_mn_o},
\begin{equation}
\label{ka_mn_p_o_xp}
\Re \ka_{mn}(p, \o; x')
= D_{mn}(p) \pi \frac{\d(\o - vp) + \d(\o + vp)}{2}
	+ \Re \ka_{mn}^{\mathrm{reg}}(p, \o; x').
\end{equation}
By similar computations as for $\Re \ka_{mn}(\o)$, starting from \eqref{ka_mn_p_o_xp_GK} and repeating the same steps but generalized in obvious ways, the results for inhomogeneous CFT can be shown to be
\begin{subequations}
\label{D_mn_p_ka_mn_p_o_xp_reg_iCFT}
\begin{equation}
\left( \begin{matrix}
	D_{11}(p) & D_{12}(p) \\
	D_{21}(p) & D_{22}(p)
\end{matrix} \right)
=
\left( \begin{matrix}
	\frac{v\ka}{\pi \b}
	& \frac{v\ka \m}{\pi \b} \\
	\frac{v\ka \m}{\pi \b}
	& \frac{\pi v c}{3 \b^3} \left[ 1 + \left( \frac{v\b p}{2\pi} \right)^2 \right]
		+ \frac{v \ka \m^2}{\pi \b}
\end{matrix} \right)
	\label{D_mn_p_iCFT}
\end{equation}
and
\begin{multline}
\left( \begin{matrix}
	\Re \ka_{11}^{\mathrm{reg}}(p, \o; x') & \Re \ka_{12}^{\mathrm{reg}}(p, \o; x') \\
	\Re \ka_{21}^{\mathrm{reg}}(p, \o; x') & \Re \ka_{22}^{\mathrm{reg}}(p, \o; x')
\end{matrix} \right) \\
=
\int_{-\io}^{\io} \dd x\,
\left( \begin{matrix}
	\frac{\ka}{2\pi \b}
	& \frac{\ka \m}{2\pi \b} \\
	\frac{\ka \m}{2\pi \b}
	& \frac{\pi c}{6 \b^3}
		\left[ 1 + \left( \frac{[\o - vp \sgn(x-x')]\b}{2\pi} \right)^2 \right]
		+ \frac{\ka \m^2}{2\pi \b}
\end{matrix} \right) I(p, \o; x, x')
	\label{ka_mn_p_o_xp_reg_iCFT}
\end{multline}
\end{subequations}
with
\begin{equation}
\label{I_omega_pxxp}
I(p, \o; x, x')
= \biggl( 1 - \frac{v}{v(x)} \biggr)
	\cos \biggl( [\o - vp \sgn(x-x')] \int_{x'}^{x} \frac{\dd x''}{v(x'')} \biggr),
\end{equation}	
where $v$ is given by the condition that $1 - v/v(x) \in L^1(\mathbb{R})$ [as explained below \eqref{I_omega}].

As in Sect.~\ref{Sec:Applications:Conductivities:Thermal_and_electrical_conductivities}, one can compute the corresponding thermal and electrical conductivities $\ka_{\mathrm{th}}(p, \o; x')$ and $\s_{\mathrm{el}}(p, \o; x')$.
Taking the ratios of their Drude weights and real regular parts when $p = 0$ imply that these alternative conductivities also satisfy the relations in \eqref{Wiedemann-Franz}.
For general $p$, the ratios instead become
\begin{equation}
\label{Wiedemann-Franz_p_xp}
\begin{aligned}
\frac{\ka}{c} 
\frac{D_{\mathrm{th}}(p)}
	{D_{\mathrm{el}}(p)}
& = \frac{\pi^2}{3 \b}
		\left[ 1 + \left( \frac{v\b p}{2\pi} \right)^2 \right], \\
\frac{\ka}{c} 
\frac{\Re \ka_{\mathrm{th}}^{\mathrm{reg}}(p, \o; x')}
	{\Re \s_{\mathrm{el}}^{\mathrm{reg}}(p, \o; x')}
& = \frac{\pi^2}{3 \b}
		\Biggl[
			1
			+
			\frac{
				\int_{-\io}^{\io} \dd x\,
				[\o - vp \sgn(x-x')]^2
				I(p, \o; x, x')
			}{
				\int_{-\io}^{\io} \dd x\,
				I(p, \o; x, x')
			}
			\left( \frac{\beta}{2\pi} \right)^2
		\Biggr].
\end{aligned}
\end{equation}
In this case, one observes that the quantum anomaly also appears for the Drude weights.


\section{Main tools}
\label{Sec:Main_tools}


In this section, we show that projective unitary representations of $\wDiff_{+}(S^1)$ and $\Map(S^1, \mathbb{R})$ on our Hilbert space can be used to flatten out $v(x)$ in \eqref{H_iCFT} as well as $v(x)\b(x)$ and $\b(x)\m(x)$ in \eqref{G_iCFT}.
This is used to prove Proposition~\ref{Proposition:Main}.

To make the presentation somewhat self-consistent, we first recall some well-known facts for representations of $\wDiff_{+}(S^1)$ and $\Map(S^1, \mathbb{R})$, see, e.g., \cite{PrSe, KhWe, GoWa1, GoWa2}.


\subsection{Projective unitary representations of diffeomorphisms}
\label{Sec:Main_tools:PUR_of_diffeos}


The Lie algebra associated with $\wDiff_{+}(S^1)$ is the infinite-dimensional algebra $\Vect(S^1)$ of smooth vector fields on $S^1$.
Any element of $\Vect(S^1)$ can be written as $\zeta(x)\partial_x$ for some smooth function $\zeta(x)$ on $S^1$ with the Lie bracket
$[\zeta_1\partial_x, \zeta_2\partial_x] = (\zeta_1'\zeta_2 - \zeta_2'\zeta_1)\partial_x$.
A central extension of $\Vect(S^1)$ by $\mathbb{R}$ is a Lie algebra
$\Vect(S^1) \oplus \mathbb{R}$
with the Lie bracket
$[(\zeta_1\partial_x, u_1), (\zeta_2\partial_x, u_2)]
= ([\zeta_1\partial_x, \zeta_2\partial_x], \omega(\zeta_1\partial_x, \zeta_2\partial_x))$
given by a 2-cocycle $\omega: \Vect(S^1) \times \Vect(S^1) \to \mathbb{R}$.%
\footnote{%
Recall that a 2-cocycle on a Lie algebra $\mathfrak{g}$ is a continuous alternating bilinear function that satisfies the cocycle identity,
$\omega(X_1, [X_2, X_3]) + \omega(X_2, [X_3, X_1]) + \omega(X_3, [X_1, X_2]) = 0$
for all $X_1, X_2, X_3 \in \mathfrak{g}$.
A 2-cocycle $\omega$ is trivial if there exists a continuous function
$\varphi: \mathfrak{g} \to \mathbb{R}$
such that $\omega = \d\varphi$ defined by $\d\varphi(X_1, X_2) = \varphi([X_1, X_2])$.%
}
A non-trivial example is the Gelfand-Fuchs cocycle  
\begin{equation}
\label{Gelfand-Fuchs_cocycle}
\omega(\zeta_1\partial_x, \zeta_2\partial_x)
= \int_{S^1} \zeta_1'\, \dd \zeta_2'
= \int_{-L/2}^{L/2} \dd x\, \zeta_1'(x) \zeta_2''(x).
\end{equation}
All non-trivial central extensions of $\Vect(S^1)$ by $\mathbb{R}$ are isomorphic, generated by the Gelfand-Fuchs cocycle with the coefficient in front of $\omega(\cdot, \cdot)$ conventionally chosen as $c/12$.
This central extension is the Virasoro algebra, denoted $\vir$.

Consider two commuting copies of $\vir$.
Their generators $L^{\pm}_{n}$ satisfy%
\footnote{%
In Fourier space, $\omega(\cdot, \cdot)$ in \eqref{Gelfand-Fuchs_cocycle} corresponds to $n^3\d_{n+m,0}$ in \eqref{Virasoro_alg_mom_space}, while the cocycle $n\d_{n+m,0}$ is trivial since $\varphi(\ell_n) = \d_{n,0}/2$ implies $\d\varphi(\ell_n, \ell_m) = \varphi([\ell_n, \ell_m]) = n\d_{n+m,0}$.%
}
\begin{equation}
\label{Virasoro_alg_mom_space}
\bigl[ L^{\pm}_{n}, L^{\pm}_{m} \bigr]
= (n-m) L^{\pm}_{n+m} + \frac{c}{12}(n^3 - n) \d_{n+m,0},
\quad
\bigl[ L^{\pm}_{n}, L^{\mp}_{m} \bigr] = 0.
\end{equation}
Defining their inverse Fourier transforms as
\begin{equation}
\label{Tpm_Fourier}
T_{\pm}(x)
= \frac{2\pi}{L^2}\sum_{n=-\io}^{\io} \ee^{\pm\frac{2\pi\ii nx}{L}}
	\left( L^{\pm}_n - \frac{c}{24} \d_{n,0} \right),
\end{equation}
one can show that the commutation relations in \eqref{Virasoro_alg_mom_space} are equivalent to \eqref{Virasoro_alg_pos_space}.

Starting with two commuting highest-weight representations of $\vir$ on the Hilbert space of any unitary CFT, they integrate to two commuting projective unitary representations of $\wDiff_{+}(S^1)$ \cite{GoWa1, GoWa2, ToLa}.
Denote these by $U_{\pm}(f)$ for $f \in \wDiff_{+}(S^1)$.
The generators of $U_{\pm}(f)$ are the components $T_{\pm}(x)$ of the energy-momentum tensor:
\begin{equation}
\label{U_pm_f_infinitesimal}
U_{\pm}(f)
= I \mp \ii \e \int_{-L/2}^{L/2} \dd x\, \zeta(x) T_{\pm}(x) + o(\e)
\end{equation}
for an infinitesimal $f(x) = x + \e \zeta(x)$ with $\zeta(x+L) = \zeta(x)$.
The phases can be chosen so that
\begin{equation}
\label{U_pm_proj_rep}
U_{\pm}(f_1) U_{\pm}(f_2) = \ee^{\pm \ii c B(f_1, f_2)/24\pi} U_{\pm}(f_1 \circ f_2),
\end{equation}
where $B: \wDiff_+(S^1) \times \wDiff_+(S^1) \to \mathbb{R}$ is a non-trivial group 2-cocycle called the Bott cocycle%
\footnote{%
For the second equality: $\log (f_1 \circ f_2)' = \log (f_1' \circ f_2) + \log f_2'$ and $\int_{S^1} \dd(\log f_2')\, \log f_2' = 0$ since $f_2'$ is periodic.%
}%
\begin{equation}
\label{Bott_cocycle}
B(f_1, f_2)
= \frac{1}{2} \int_{S^1} \log (f_1 \circ f_2)'\, \dd(\log f_2')
= \frac{1}{2} \int_{-L/2}^{L/2} \dd x\, [\log f_2'(x)]'\, \log [f_1'(f_2(x))].
\end{equation}
The central extension of $\wDiff_+(S^1)$ corresponding to \eqref{Bott_cocycle} is called the Virasoro-Bott group, and one can show that $B(\cdot, \cdot)$ reduces to $\omega(\cdot, \cdot)$ in \eqref{Gelfand-Fuchs_cocycle} for infinitesimal diffeomorphisms.

For our purposes, only the adjoint action (actually the coadjoint action) of $U_{\pm}(f)$ is needed.
One can show that
\begin{equation}
\label{U_Tpm_Ui}
U_{\pm}(f) T_{\pm}(x) U_{\pm}(f)^{-1}
= f'(x)^2 T_{\pm}(f(x)) - \frac{c}{24\pi} \{ f(x), x \},
\quad
U_{\pm}(f) T_{\mp}(x) U_{\pm}(f)^{-1}
= T_{\mp}(x)
\end{equation}
with $\{ f(x), x \}$ in \eqref{Schwarzian_derivative}.
These can be shown to be consistent with \eqref{U_pm_proj_rep}.
Note also that \eqref{U_Tpm_Ui} implies \eqref{U_Tpm_Ui_prerequisites}.


\subsection{Projective unitary representations of smooth maps}
\label{Sec:Main_tools:PUR_of_smooth_maps}


The group $\Map(S^1, \mathbb{R})$ is an example of a loop group.
We are interested in the $\mathfrak{u}(1)$-current algebra obtained as the central extension by $\mathbb{R}$ of the corresponding loop algebra.
As for $\vir$, we have two commuting copies of the $\mathfrak{u}(1)$-current algebra whose generators $J^{\pm}_{n}$ satisfy
\begin{equation}
\label{Virasoro_and_u1_current_alg_mom_space}
[J^\pm_n,J^\pm_m] = \ka n\d_{n+m,0},
\quad
[J^\pm_n,J^\mp_m] = 0,
\quad
[L_n^\pm,J^\pm_m] = -mJ^\pm_{n+m},
\quad
[L_n^\pm,J^\mp_m] = 0.
\end{equation}
Their inverse Fourier transforms are
\begin{equation}
\label{Jpm_Fourier}
J_{\pm}(x)
= \frac{1}{L} \sum_{n=-\io}^{\io}
	\ee^{\pm\frac{2\pi\ii nx}{L}} J^{\pm}_{n}.
\end{equation}
As before, using \eqref{Tpm_Fourier} and \eqref{Jpm_Fourier}, one can show that the commutation relations in
\eqref{Virasoro_alg_mom_space} and \eqref{Virasoro_and_u1_current_alg_mom_space}
are equivalent to
\eqref{Virasoro_alg_pos_space} and \eqref{Virasoro_and_u1_current_alg_pos_space}.

For any unitary CFT with two commuting highest-weight representations of the $\mathfrak{u}(1)$-current algebra, the latter integrate to two commuting projective unitary representation $V_{\pm}(h)$ of $h \in \Map(S^1, \mathbb{R})$ on the Hilbert space of the theory \cite{GoWa1, ToLa}.
The generators of these are the components $J_{\pm}(x)$ of the $\mathrm{U}(1)$ current:
\begin{equation}
\label{V_pm_h_infinitesimal}
V_{\pm}(h)
= I \mp \ii \eps \int_{-L/2}^{L/2} \dd x\, \xi(x) J_{\pm}(x) + o(\eps)
\end{equation}
for an infinitesimal $h(x) = \eps\xi(x)$ with $\xi(x+L) = \xi(x)$.

To complete the picture, consider the semi-direct product $\Map(S^1, \mathbb{R}) \rtimes \wDiff_{+}(S^1)$ with elements $(h, f)$ and group operation
$(h_1, f_1) \cdot (h_2, f_2) = (f_2^{\ast}h_1 + h_2, f_1 \circ f_2)$,
where $f^{\ast}h = h \circ f$ is the pullback of $h$ by $f$.
Denote by $U_{\pm}(h, f)$ two commuting projective unitary representations of $(h, f) \in \Map(S^1, \mathbb{R}) \rtimes \wDiff_{+}(S^1)$.
Since $(h, f) = (0, f) \cdot (h, \Id)$,
\begin{equation}
U_{\pm}(h, f)
= U_{\pm}(0, f) U_{\pm}(h, \Id)
= U_{\pm}(f) V_{\pm}(h),
\end{equation}
using $U_{\pm}(0, f) = U_{\pm}(f)$ and $U_{\pm}(h, \Id) = V_{\pm}(h)$.
As before,
\begin{equation}
\label{U_pm_V_pm_proj_rep}
U_{\pm}(h_1, f_1) U_{\pm}(h_2, f_2)
= \ee^{\pm \ii C((h_1, f_1),(h_2, f_2))}
	U_{\pm}((h_1, f_1) \cdot (h_2, f_2)),
\end{equation}
where $C((h_1, f_1),(h_2, f_2))$ is a more general group 2-cocycle that includes both the Bott cocycle and the corresponding 2-cocycle for $\Map(S^1, \mathbb{R})$, see, e.g., \cite{PrSe, KhWe, GoWa1, GoWa2}.

As before, we are only interested in the (co)adjoint action.
One can show that
\begin{equation}
\label{U_Jpm_Ui}
U_{\pm}(f) J_{\pm}(x) U_{\pm}(f)^{-1}
= f'(x) J_{\pm}(f(x)),
\quad
U_{\pm}(f) J_{\mp}(x) U_{\pm}(f)^{-1}
= J_{\mp}(x)
\end{equation}
and
\begin{subequations}
\label{V_Tpm_Jpm_Vi}
\begin{align}
V_{\pm}(h) T_{\pm}(x) V_{\pm}(h)^{-1}
& = T_{\pm}(x) + h'(x) J_{\pm}(x) + \frac{\ka}{4\pi}h'(x)^2,
& V_{\pm}(h) T_{\mp}(x) V_{\pm}(h)^{-1}
& = T_{\mp}(x),
		\label{V_Tpm_Vi} \\
V_{\pm}(h) J_{\pm}(x) V_{\pm}(h)^{-1}
& = J_{\pm}(x) + \frac{\ka}{2\pi}h'(x),
& V_{\pm}(h) J_{\mp}(x) V_{\pm}(h)^{-1}
& = J_{\mp}(x).
		\label{V_Jpm_Vi}
\end{align}
\end{subequations}
The above imply \eqref{U_Jpm_Ui_prerequisites} and \eqref{V_Tpm_Jpm_Vi_prerequisites} and can be shown to be consistent with \eqref{U_pm_V_pm_proj_rep}.

\begin{remark}[Sugawara construction]
\label{Remark:Sugawara_construction}
Given a $\mathfrak{u}(1)$-current algebra, it is possible to construct a corresponding Virasoro algebra as follows:
\begin{equation}
\label{Sugawara_construction}
2\ka L^{\pm}_{n} = \sum_{m} \! \fermionWick{ J^{\pm}_{n-m} J^{\pm}_{m} } \!.
\end{equation}
This is the so-called Sugawara construction, cf., e.g., \cite{KnZa, Zamo}.
Note that Kronig's identity in Remark~\ref{Remark:Bosonization} follows as a special case.
For a CFT with Virasoro algebra given by the Sugawara construction, any primary field $\Phi$ [satisfying \eqref{U_Phi_Ui_prerequisites} and \eqref{V_Phi_Vi_prerequisites}] must have weights that obey
\begin{equation}
\label{Sugawara_condition}
2\ka \D^\pm_{\Phi} = (\tau^\pm_{\Phi})^2.
\end{equation}
(To see this, let $|\Phi\rangle$ be the primary state associated with $\Phi$, then $L^{\pm}_{0} |\Phi\rangle = \D^{\pm}_{\Phi} |\Phi\rangle$ and $J^{\pm}_{0} |\Phi\rangle = \tau^{\pm}_{\Phi} |\Phi\rangle$, which together with \eqref{Sugawara_construction} implies \eqref{Sugawara_condition}, cf., e.g., \cite{Meln}.)
In addition, $c = \dim(\mathfrak{u}(1)) = 1$ for such a CFT, cf.\ \cite{KnZa}.
Note that these properties apply to the examples in Sect.~\ref{Sec:Prerequisites:Examples}.
\end{remark}
%


\subsection{Proof of Proposition~\texorpdfstring{\ref{Proposition:Main}}{}}
\label{Sec:Main_tools:Proof_of_result}


The proof relies on three lemmas.

\begin{lemma}
\label{Lemma:fgh}
Given $v(x)$, $\b(x)$, and $\m(x)$ in Proposition~\ref{Proposition:Main}, then $f$ in \eqref{f_v_x} and $g$ in \eqref{g_vbeta_x} define elements in $\wDiff_{+}(S^1)$ and $h$ in \eqref{h_vmu_x} defines an element in $\Map(S^1, \mathbb{R})$.
\end{lemma}
\begin{proof}
Since $v(x) > 0$, $f$ is an orientation-preserving diffeomorphism (by the inverse function theorem), and the choice of $v_0$ then implies that $f \in \wDiff_{+}(S^1)$.
Similarly, one can show that $g \in \wDiff_+(S^1)$ and that $h \in \Map(S^1, \mathbb{R})$.
\end{proof}

It follows from Lemma~\ref{Lemma:fgh} and Sect.~\ref{Sec:Main_tools:PUR_of_diffeos} that $U(f) = U_{+}(f)U_{-}(f)$ and $U(g) = U_{+}(g)U_{-}(g)$ define projective unitary representations of $f$ and $g$ on the Hilbert space of any unitary CFT.
In addition, it follows from Sect.~\ref{Sec:Main_tools:PUR_of_smooth_maps} that $V(h) = V_{+}(h)V_{-}(h)$ defines a projective unitary representation of $h$ on the Hilbert space of any unitary CFT with a $\mathfrak{u}(1)$-current algebra.
These representations of $f$, $g$, and $h$ given by \eqref{fgh} are used in what follows.

\begin{lemma}
\label{Lemma:H_iCFT}
Given $H$ in \eqref{H_iCFT} and $f$ in \eqref{f_v_x}, then
\begin{equation}
\label{Uf_H_Ufi}
U(f) H U(f)^{-1}
= H_{0} - \int_{-L/2}^{L/2} \dd x\, \frac{\cS(x)}{v(x)}
\end{equation}
with $H_{0}$ in \eqref{H_sCFT} and $\cS(x)$ in \eqref{cS}.
Moreover, given \eqref{H_dynamics} and $\tilde{x}^{\pm}$ in \eqref{x_pm_tilde}, then
\begin{subequations}
\begin{align}
T_{\pm}(x; t)
& = \biggl( \frac{v(\tilde{x}^{\mp})}{v(x)} \biggr)^2 T_{\pm}(\tilde{x}^{\mp})
			+ \frac{\cS(\tilde{x}^{\mp}) - \cS(x)}{2 v(x)^2},
		\label{HtTpmHt} \\
J_{\pm}(x; t)
& = \frac{v(\tilde{x}^{\mp})}{v(x)} J_{\pm}(\tilde{x}^{\mp}),
		\label{HtJpmHt} \\
\Phi(x; t)
& = \biggl( \frac{v(\tilde{x}^-)}{v(x)} \biggr)^{\D^+_{\Phi}}
		\biggl( \frac{v(\tilde{x}^+)}{v(x)} \biggr)^{\D^-_{\Phi}}
		\Phi(\tilde{x}^-, \tilde{x}^+),
		\label{HtPhiHt}
\end{align}
\end{subequations}
where the latter is for any Virasoro primary field $\Phi$ with conformal weights $(\D^+_{\Phi}, \D^-_{\Phi})$.
\end{lemma}
\begin{proof}
Setting $y = f(x)$, it follows from \eqref{U_Tpm_Ui_prerequisites} that
\begin{subequations}
\begin{align}
U(f) T_{\pm}(x) U(f)^{-1}
& = T_{\pm}(y) \left( \frac{\dd y}{\dd x} \right)^2 - \frac{c}{24\pi} \{ y, x \},
		\label{U_Tpm_Ui_y} \\
U(f)^{-1} T_{\pm}(y) U(f)
& = T_{\pm}(x) \left( \frac{\dd x}{\dd y} \right)^2 - \frac{c}{24\pi} \{ x, y \},
		\label{Ui_Tpm_U_y}
\end{align}
\end{subequations}
where we used $\{ y, x \} = - (\dd y/\dd x)^2 \{ x, y \}$ \cite{FMS}.
Since $\dd x/\dd y = v(x)/v_{0}$, \eqref{U_Tpm_Ui_y} implies
\begin{equation}
\label{UintTpmUi}
U(f) \left( \int_{-L/2}^{L/2} \dd x\, v(x) T_{\pm}(x) \right) U(f)^{-1}
= \int_{-L/2}^{L/2} \dd y\, v_{0} T_{\pm}(y)
		- \frac{c}{24\pi} \int_{-L/2}^{L/2} \dd x\, v(x) \{ y, x \}.
\end{equation}
Thus, conjugating $H$ in \eqref{H_iCFT} with $U(f)$ yields \eqref{Uf_H_Ufi} with the Schwarzian-derivative contribution $\cS(x) = c v(x)^2 \{ f(x), x \} /12 \pi$, where the latter can equivalently be written as in \eqref{cS} using $f$ in \eqref{f_v_x}.

To prove \eqref{HtTpmHt}, note that \eqref{Uf_H_Ufi} together with \eqref{U_Tpm_Ui_y} and \eqref{Ui_Tpm_U_y} implies
\begin{align}
T_{\pm}(x; t)
& = U(f)^{-1} \ee^{\ii H_{0} t} U(f)
		T_{\pm}(x)
		U(f)^{-1} \ee^{-\ii H_{0} t} U(f) \nonumber \\
& = f'(x)^2
		\left[
			T_{\pm}(\tilde{x}^{\mp}) \bigl[ (f^{-1})'(y \mp v_{0}t) \bigr]^2
			- \frac{c}{24\pi} \{ f^{-1}(y \mp v_{0}t), y \mp v_{0}t \}
		\right]
		- \frac{c}{24\pi} \{ y, x \}.
		\label{Tpm_xt_in_proof}
\end{align}
Again, since ${\dd y}/{\dd x} = {v_{0}}/{v(x)}$, one can show that $(f^{-1})'(y \mp v_{0}t) = v(\tilde{x}^{\mp})/v_{0}$ as well as
\begin{equation}
\frac{c}{24 \pi} \{ y, x \}
= \frac{\cS(x)}{2 v(x)^2},
\quad
\frac{c}{24 \pi} \{ f^{-1}(y \mp v_{0}t), y \mp v_{0}t \}
= - \frac{\cS(\tilde{x}^{\mp})}{2 v_{0}^2},
\end{equation}
where we used the formula for the Schwarzian derivative below \eqref{Ui_Tpm_U_y} in the last equation.
Combining this with \eqref{Tpm_xt_in_proof} yields the desired result.

The proofs of \eqref{HtJpmHt} and \eqref{HtPhiHt} follow analogously by using the transformations in \eqref{U_Jpm_Ui_prerequisites} and \eqref{U_Phi_Ui_prerequisites} together with their inverses.
\end{proof}
\begin{lemma}
\label{Lemma:G_iCFT}
Given $G$ in \eqref{G_iCFT} and $g$ and $h$ in \eqref{fgh}, then
\begin{equation}
\label{Ug_Vh_H_Vhi_Ugi}
U(g) V(h) \ee^{-G} V(h)^{-1} U(g)^{-1}
= \ee^{-\b_{0} (H_{0} - \m_{0}Q_{0}) + \mathrm{const}}
\end{equation}
with $H_{0}$ and $Q_{0}$ in \eqref{H_sCFT}.
\end{lemma}
\begin{proof}
The derivation is analogous to that of Lemma~\ref{Lemma:H_iCFT} using the transformation rules in \eqref{U_Tpm_Ui_prerequisites}, \eqref{U_Jpm_Ui_prerequisites}, and \eqref{V_Tpm_Jpm_Vi_prerequisites}.
\end{proof}

The results in Proposition~\ref{Proposition:Main} follow straightforwardly from the lemmas above and by computing
$\tilde{\cO}(x; t) =  U(g) V(h) \cO(x; t) V(h)^{-1} U(g)^{-1}$
for $\cO(x; t)$ equal to
$T_{\pm}(x; t)$ using \eqref{U_Tpm_Ui_prerequisites} and \eqref{V_Tpm_Vi_prerequisites},
$J_{\pm}(x; t)$ using \eqref{U_Jpm_Ui_prerequisites} and \eqref{V_Jpm_Vi_prerequisites},
and
$\Phi(x; t)$ using \eqref{U_Phi_Ui_prerequisites} and \eqref{V_Phi_Vi_prerequisites}. 


\section{Concluding remarks}
\label{Sec:Concluding_remarks}


In this paper, we defined inhomogeneous conformal field theory as a 1+1-dimensional CFT with a smooth position-dependent velocity $v(x)$ explicitly breaking translation invariance.
We showed how exact analytical results for such models out of equilibrium can be obtained using projective unitary representations of diffeomorphisms and smooth maps.
In particular, we derived explicit formulas for the inhomogeneous dynamics and for the thermal and electrical conductivities, which generalize well-known results for standard CFT.

The conductivities were computed in two ways:
The first based on a dynamical approach and the second using a Green-Kubo formula that we derived.
We stress that the equivalence between these two is non-trivial.
The first is fully dynamical, here based on a quantum quench from initial states with inverse-temperature and chemical-potential profiles given by a kink-like function $W(x)$, while the second is based on equilibrium current-current correlation functions.
On general grounds, they must be equivalent, cf.\ \cite{Spo2}, but verifying this is not straightforward; see, e.g., \cite{IlPe} for a discussion of such an equivalence between dynamically computed Drude weights and Green-Kubo-type formulas for one-dimensional lattice fermions and \cite{HeTe, MPT} for a review and recent related results for gapped quantum systems.
In particular, when deriving \eqref{D_th_ka_th_reg_iCFT} using the dynamical approach, it becomes clear that the factor $1 + ({\o\b}/{2\pi})^2$ is due to a quantum anomaly originating from a Schwinger term [see \eqref{Virasoro_alg_pos_space}] which appears ubiquitously in CFT.
This observation would not be evident from the Green-Kubo approach, even if the final expressions are the same.
As discussed in Sect.~\ref{Sec:Introduction}, the physical significance of this quantum anomaly includes a generalization of the Wiedemann-Franz law within inhomogeneous CFT to finite times.
Moreover, we note that $D_{\mathrm{th}}$ in \eqref{D_th_ka_th_reg_iCFT} and $D_{\mathrm{el}}$ in \eqref{D_el_s_el_reg_iCFT} are the same universal results as in standard CFT, see, e.g., \cite{GLM}, while $\Re \ka_{\mathrm{th}}^{\mathrm{reg}}(\o)$ and $\Re \s_{\mathrm{el}}^{\mathrm{reg}}(\o)$ are non-universal since they depend on the details of the functions $W(x)$ and $v(x)$.
Finally, we also computed conductivities $\ka_{\mathrm{th}}(p, \o; x')$ and $\s_{\mathrm{el}}(p, \o; x')$ as the linear responses to perturbations at position $x'$, which exhibit the same quantum anomaly but whose real regular parts are independent of the details of the perturbations.

This paper both generalizes and lays the mathematical foundations for \cite{LaMo2}, where we explicitly showed that a random $v(x)$ leads to the emergence of diffusive heat transport after averaging over the randomness.
Two approaches were given in \cite{LaMo2}:
The first based on random inhomogeneous dynamics and the second using the explicit expression for $\ka_{\mathrm{th}}(\omega)$.
Regarding the first approach, we emphasize that the inhomogeneous dynamics is encoded in the generalized light-cone coordinates $\tilde{x}^{\pm} = \tilde{x}^{\pm}_{t}(x)$ satisfying \eqref{dxpm_dx} for fixed $v(x)$.
Even if our results depend on these coordinates in explicit ways, such as in \eqref{cEcJ_rj_iCFT_neq}, extracting information remains complicated.
In particular, it is difficult to compute the average $\mathbb{E}[\cdot]$ for a $v(x)$ given by a random function.
Inspired by wave propagation in random media \cite{Ish}, this was investigated in \cite{LaMo2} by instead directly studying the random partial differential equations that
the expectations of the energy density and heat current operators satisfy.
In the second approach, we subtracted a Drude peak as in \eqref{ka_th} and \eqref{D_th_ka_th_reg_iCFT} with $v^{-1} = \mathbb{E}[v(x)^{-1}]$ from $\ka_{\mathrm{th}}(\omega)$, which after averaging gave an explicit expression for $\mathbb{E}[\Re \ka_{\mathrm{th}}^{\mathrm{reg}}(\omega)]$.
The limit $\omega \to 0$ could safely be taken in this expression and yielded a non-zero value.
One interpretation is that the normal diffusive contribution in \cite{LaMo2} is an emergent phenomenon in the sense that it reflects a lack of knowledge about mesoscopic details, which manifests itself as diffusion on larger scales after averaging.
It would be interesting to better understand this, including also for the anomalous diffusive contribution in \cite{LaMo2}, as well as how to interpret the results for a fixed velocity $v(x)$.

As a final remark, we note that emergent Euler-scale hydrodynamic results, see, e.g., \cite{Doyon:Lecture_notes}, can be derived from the exact infinite-volume results in Sect.~\ref{Sec:Applications}.
This scaling involves replacing $v(\cdot)$, $\beta(\cdot)$, $\mu(\cdot)$, and $(x, t)$ by $v(\cdot/\lambda)$, $\beta(\cdot/\lambda)$, $\mu(\cdot/\lambda)$, and $(\lambda  x, \lambda t)$, respectively, for some $\lambda > 0$ and taking the limit $\lambda \to \io$; cf.\ \cite{Moo:GHD-NLLM} where this was studied for the non-local Luttinger model.
[Naturally, for our generalized light-cone coordinates, this scaling entails that $\tilde{x}^{\pm}$ are replaced by $\lambda \tilde{x}^{\pm}$.]
For the results in Sect.~\ref{Sec:Applications:Densities_and_currents}, the effects of sending $\lambda \to \io$ are that $\cS(x)$ and $\cT(x)$ in \eqref{cEcJ_rj_iCFT_neq} and \eqref{F_x_G_x_iCFT_neq} disappear, while the other terms remain unaffected.
It would be interesting to further explore connections with Euler-scale hydrodynamics for inhomogeneous dynamics.

\paragraph{Acknowledgements:}
I want to express my gratitude to Edwin Langmann for valuable suggestions and collaboration at an early stage of this work.
I also want to thank Eddy Ardonne, Jens Bardarson, Benjamin Doyon, Krzysztof Gaw\k{e}dzki, Gian Michele Graf, Bastien Lapierre, Jonathan Lindgren, Jouko Mickelsson, Blagoje Oblak, and Herbert Spohn for helpful discussions and remarks.
Financial support from the Wenner-Gren Foundations (No.\ WGF2019-0061) is gratefully acknowledged.


\begin{appendices}


\section{Linear-response theory}
\label{App:Linear-response_theory}


In this appendix, we review linear-response theory for 1+1-dimensional systems in the case of an arbitrary number of conserved charges.%
\footnote{%
Cf.\ Footnote~\ref{Footnote:Higher_conserved_charges}.%
}
In particular, the formulas in \eqref{ka_mn_o_dynamic} and \eqref{ka_mn_o_GK} for the conductivities are derived and shown to give the same result.%
\footnote{%
The presentation here is an updated version of the corresponding appendix in \cite{Moo}, which in turn was based on private correspondence with Krzysztof Gaw\k{e}dzki when working on \cite{GLM}.%
}


\subsection{Linear response in closed quantum systems}
\label{App:Linear-response_theory:Linear_response_in_closed_quantum_systems}


Let $H_{\mathrm{sys}}$ be the Hamiltonian [not necessarily the one in \eqref{H_iCFT}] for some system that, in general, can be inhomogeneous.
We introduce the time-dependent Hamiltonian
$H_{\mathrm{sys}}(\boldsymbol{\la}(s))
= H_{\mathrm{sys}} - \sum_{n \geq 1} \la_{n}(s) V_{n}$,
$\boldsymbol{\la}(s) = (\la_{1}(s), \la_{2}(s), \ldots)$,
where $\la_{n}(s)$ are functions of time $s$ and $V_{n}$ are perturbations ($n = 1, 2, \ldots$).
Suppose that the system is in the equilibrium initial state
$\hat{\rho} = \ee^{-\b H_{\mathrm{sys}}}/\Tr [ \ee^{-\b H_{\mathrm{sys}}} ]$ at times $t < t_{0}$ and consider the time-evolved state $\hat{\rho}_{\boldsymbol{\la}(\cdot)}(t)$ under $H_{\mathrm{sys}}(\boldsymbol{\la}(s))$ for $t > t_{0}$,
\begin{equation}
\label{rho_t_lambda}
\hat{\rho}(t, \boldsymbol{\la}(\cdot))
= \overleftarrow{\cT} \!
	\exp \biggl(
		-\ii \int_{t_{0}}^{t} \dd s\, H_{\mathrm{sys}}(\boldsymbol{\la}(s))
	\biggr)
	\hat{\rho} \,
	\overrightarrow{\cT} \!
	\exp \biggl(
		\ii \int_{t_{0}}^{t} \dd s\, H_{\mathrm{sys}}(\boldsymbol{\la}(s))
	\biggr)
\end{equation}
with time ordering $\overrightarrow{\cT}$ $(\overleftarrow{\cT})$ such that time increases (decreases) from left to right.
For operators $O_{m}$ ($m = 1, 2, \ldots$), define the response functions $R_{mn}(t, s)$ for $s > t_0$ by \cite{Kubo}
\begin{equation}
\label{R_mn_t_s}
R_{mn}(t, s)
=	\frac{\d}{\d \la_{n}(s)}
	\Bigl(
		\Tr \bigl[ O_{m} \hat{\rho}(t, \boldsymbol{\la}(\cdot)) \bigr]
	\Bigr) \Big|_{\boldsymbol{\la}(\cdot) = {\bf 0}}.
\end{equation}
Without loss of generality, we can set $t_{0} = -\io$.
Clearly, $R_{mn}(t, s) = 0$ for $s > t$, while
\begin{equation}
\label{R_mn_t_s_nonzero}
R_{mn}(t, s)
=	\Tr \Bigl[
		O_{m} \ee^{-\ii(t-s)H_{\mathrm{sys}}}
		[\ii V_{n}, \hat{\rho}] \ee^{\ii(t-s)H_{\mathrm{sys}}}
	\Bigr]
=	\Tr \Bigl[ O_{m}(t-s) [\ii V_{n}, \hat{\rho}] \Bigr]
\end{equation}
for $s \leq t$, where $O_{m}(t) = \ee^{\ii H_{\mathrm{sys}}t} O_{m} \ee^{-\ii H_{\mathrm{sys}}t}$.
Since $\bigl[ H_{\mathrm{sys}} - \la_{n} V_{n}, \ee^{-\b(H_{\mathrm{sys}} - \la_{n} V_{n})} \bigr] = 0$,
\begin{equation}
0
= \frac{\partial}{\partial\la_{n}}
		\Bigl(
			\Bigl[
				H_{\mathrm{sys}} - \la_{n} V_{n},
					\ee^{-\b(H_{\mathrm{sys}} - \la_{n} V_{n})}
			\Bigr]
		\Bigr)
		\Big|_{\la_{n} = 0}
= - \Bigl[ V_{n}, \ee^{-\b H_{\mathrm{sys}}} \Bigr]
		+ \biggl[
				H_{\mathrm{sys}},
					\frac{\partial}{\partial\la_{n}}
					\Bigl( \ee^{-\b(H_{\mathrm{sys}} - \la_{n} V_{n})} \Bigr) \Big|_{\la_{n} = 0}
			\biggr],
\end{equation}
where
\begin{equation}
\label{dl_ebHlV_l0}
\frac{\partial}{\partial\la_{n}}
	\Bigl( \ee^{-\b(H_{\mathrm{sys}} - \la_{n} V_{n})} \Bigr) \Big|_{\la_{n} = 0}
= \int_{0}^{\b} \dd \t\, \ee^{-\t H_{\mathrm{sys}}} V_{n} \ee^{(\t-\b)H_{\mathrm{sys}}}
= \int_{0}^{\b} \dd \t\, V_{n}(\ii\t) \ee^{-\b H_{\mathrm{sys}}},
\end{equation}
which means that
\begin{equation}
\bigl[ \ii V_{n}, \hat{\rho} \bigr]
= \int_{0}^{\b} \dd \t\, \ii \bigl[ H_{\mathrm{sys}}, V_{n}(\ii\t) \bigr] \hat{\rho}
= \int_{0}^{\b} \dd \t\, \partial\pdag_{\ii\t} V_{n}(\ii\t) \hat{\rho}.
\end{equation}
Inserting this into \eqref{R_mn_t_s_nonzero} yields
\begin{equation}
R_{mn}(t, s)
= \int_{0}^{\b} \dd \t\,
	\Tr \Bigl[ O_{m}(t-s) \partial\pdag_{\ii\t} V_{n}(\ii\t) \hat{\rho} \Bigr]
= R_{mn}(t-s, 0).
\end{equation}
In conclusion, defining $R_{mn}(t) = R_{mn}(t, 0)$, we have shown that
\begin{equation}
\label{R_mn_t}
R_{mn}(t)
= \int_{0}^{\b} \dd \t\,
	\bigl\langle O_{m}(t) \partial\pdag_{\ii\t} V_{n}(\ii\t) \bigr\rangle_{\b}
\end{equation}
with $\langle \cdots \rangle_{\b} = \Tr \bigl[ (\cdots) \hat{\rho} \bigr]$ and $\hat{\rho}$ defined above \eqref{rho_t_lambda}.


\subsection{Linear response from quench dynamics}
\label{App:Linear-response_theory:Linear_response_from_quench_dynamics}


Let
$\hat{\rho}_{\boldsymbol{\la}}
= \ee^{-\b(H_{\mathrm{sys}}-\sum_{n \geq 1} \la_{n} V_{n})}
	/ \Tr [ \ee^{-\b(H_{\mathrm{sys}}-\sum_{n \geq 1} \la_{n} V_{n})} ]$,
$\boldsymbol{\la} = (\la_{1}, \la_{2}, \ldots)$,
define the initial state and consider
\begin{equation}
\label{spp_expectation}
\bigl\langle O_{m}(t) \bigr\rangle_{\boldsymbol{\la}}
= \Tr \bigl[ O_{m}(t) \hat{\rho}_{\boldsymbol{\la}} \bigr]
\end{equation}
for operators $O_{m}(t) = \ee^{\ii H_{\mathrm{sys}}t} O_{m} \ee^{-\ii H_{\mathrm{sys}}t}$ evolving under $H_{\mathrm{sys}}$ for $t > 0$.
This is a quantum quench changing the Hamiltonian from
$H_{\mathrm{sys}} - \sum_{n \geq 1} \la_{n} V_{n}$ to $H_{\mathrm{sys}}$ at $t = 0$.
The aim is to show that the response functions for changes in $\la_{n}$ defined above can be expressed as
\begin{equation}
\label{R_mn_dynamic}
R_{mn}(t) 
= -	\frac{\partial}{\partial \la_{n}}
		\Bigl(
			\partial_{t} \bigl\langle O_{m}(t) \bigr\rangle_{\boldsymbol{\la}}
		\Bigr) \Big|_{\boldsymbol{\la} = {\bf 0}}.
\end{equation}
Indeed,
$\partial_{t} \bigl\langle O_{m}(t) \bigr\rangle_{\boldsymbol{\la}}
= \Tr \bigl[ \ii[H_{\mathrm{sys}}, O_{m}(t)] \hat{\rho}_{\boldsymbol{\la}} \bigr]$,
and thus,
\begin{align}
\frac{\partial}{\partial \la_{n}}
	\Bigl(
		\partial_{t} \langle O_{m}(t) \rangle_{\boldsymbol{\la}}
	\Bigr) \Big|_{\boldsymbol{\la} = {\bf 0}}
& = \frac{1}{\Tr \ee^{-\b H_{\mathrm{sys}}}}
		\Tr
		\biggl[
			\ii[H_{\mathrm{sys}}, O_{m}(t)]
			\frac{\partial}{\partial \la_{n}}
			\Bigl( \ee^{-\b(H_{\mathrm{sys}}-\la_{n} V_{n})} \Bigr) \Big|_{\la_{n} = 0}
		\biggr] \nonumber \\
& = \frac{1}{\Tr \ee^{-\b H_{\mathrm{sys}}}}
		\Tr
		\biggl[
			\ii[H_{\mathrm{sys}}, O_{m}(t)]
			\int_{0}^{\b} \dd \t\, V_{n}(\ii\t) \ee^{-\b H_{\mathrm{sys}}}
		\biggr] \nonumber \\
& = - \int_{0}^{\b} \dd \t\,
			\bigl\langle O_{m}(t) \partial\pdag_{\ii\t} V_{n}(\ii\t) \bigr\rangle_{\b}
	= - R_{mn}(t),
\end{align}
where we used \eqref{dl_ebHlV_l0} in the second equality and \eqref{R_mn_t} in the last.

We stress that \eqref{R_mn_dynamic} is a general result saying that the response functions $R_{mn}(t)$ obtained from the equilibrium correlation function in \eqref{R_mn_t} can equivalently be computed from the dynamics following a quantum quench.


\subsection{Conductivity matrix}
\label{App:Linear-response_theory:Conductivity_matrix}


In what follows the above is specialized to kink-like profiles as in Sect.~\ref{Sec:Applications:Conductivities}.
Let $\boldsymbol{\m} = (\m_{1}, \m_{2}, \ldots)$, where we recall that $\m_{1} = \b\m$ and $\m_{2} = -\b$.
We identify
$H_{\mathrm{sys}}$ with $Q_{2} - \sum_{n \neq 2} (\m_n/\b) Q_{n}$ using $Q_{2} = H$ in \eqref{H_iCFT} and $Q_{1} = Q$ in \eqref{Q_iCFT}, and we pick
$V_{n} = \int_{-L/2}^{L/2} \dd x\, W(x) \rho_{n}(x)$ and
$O_{m} = \int_{-L/2}^{L/2} \dd x\, j_{m}(x)$.
We recall that the latter are precisely the total currents and that $W(x)$ describes an overall kink-like profile such that $\lim_{x\to\mp\io} W(x) = \pm 1/2$ in the infinite volume.
A specific example of such a function $W(x)$ used in \cite{GLM} is
\begin{equation}
\label{W_x}
W(x) = W_0 \biggl( \frac{L}{2\pi} \sin \biggl( \frac{2\pi x}{L} \biggr) \biggr),
\quad
W_0(x) = - \frac{1}{2} \tanh \biggl( \frac{x}{\d_{W}} \biggr)
\end{equation}
for $\d_{W} > 0$, which satisfies our periodic boundary conditions for finite $L$.

The sign convention for $R_{mn}(t, s)$ in \eqref{R_mn_t_s} is so that an overall positive gradient [i.e., a negative $\la_{n}$ since $W(x)$ goes from $1/2$ to $-1/2$ in the infinite volume] corresponds to $R_{mn}(t, s)$ positive.
Since such a gradient induces a negative current, we define the conductivities as
\begin{equation}
\label{conductivity_def}
\ka_{mn}(t)
= \lim_{L\to\io} - \frac{1}{\beta} R_{mn}(t)
\end{equation}
in the thermodynamic limit.

\begin{proof}[Proof of \eqref{ka_mn_o_dynamic} and \eqref{ka_mn_o_GK}]
To prove the first equation, if we identify $\d\m_{n}$ with $\b \la_{n}$, it follows from \eqref{R_mn_dynamic} and \eqref{conductivity_def} that
\begin{equation}
\label{ka_mn_o_dynamic_appendix}
\ka_{mn}(t)
= \lim_{L\to\io} \frac{\partial}{\partial (\d\m_{n})}
	\Bigl(
		\partial_{t} \langle O_{m}(t) \rangle_{\boldsymbol{\la}}
	\Bigr) \Big|_{\boldsymbol{\la} = {\bf 0}}
= \lim_{L\to\io} \frac{\partial}{\partial (\d\m_{n})}
	\int_{-L/2}^{L/2} \dd x\,
	\Bigl(
		\partial_{t} \langle j_{m}(x; t) \rangle_{\mathrm{neq}}
	\Bigr) \Big|_{\boldsymbol{\d\m} = {\bf 0}}
\end{equation}
using the definition of the expectation in \eqref{spp_expectation}.
To prove the second, using the expression for $V_{n}$ above and the continuity equation
$\partial_t \rho_{n} + \partial_x j_{n} = 0$, we obtain
\begin{equation}
\partial\pdag_{\ii\t} V_{n}(\ii\t)
= - \int_{-L/2}^{L/2} \dd x\, W(x) \partial_x j_{n}(x, \ii\t).
\end{equation}
This together with \eqref{R_mn_t} and \eqref{conductivity_def} yields
\begin{align}
\label{ka_mn_o_GK_appendix}
\ka_{mn}(t)
& = \frac{1}{\b} \int_{0}^{\b} \dd \t
		\int_{-\io}^{\io} \dd x \int_{-\io}^{\io} \dd x'
		\bigl\langle j_{m}(x; t) W(x')
			\partial_{x'} j_{n}(x'; \ii\t) \bigr\rangle_{\b}^{\io} \nonumber \\
& = \frac{1}{\b} \int_{0}^{\b} \dd \t
		\int_{-\io}^{\io} \dd x'\, \partial_{x'} [-W(x')]
		\int_{-\io}^{\io} \dd x \,
		\bigl\langle j_{m}(x; t) j_{n}(x'; \ii\t) \bigr\rangle_{\b}^{c, \io}
\end{align}
with $\langle \cdots \rangle_{\b}$ defined below \eqref{R_mn_t}, where we used integration by parts, assumed that the current-current correlation functions decay rapidly for large separations, and used that the connected part is the only non-zero contribution since $\langle j_{m}(x; t) \rangle_{\b}^{\io} = 0$.

The results in \eqref{ka_mn_o_dynamic} and \eqref{ka_mn_o_GK} follow from the above by passing to the frequency domain using
$\ka_{mn}(\o)
= \int_{0}^{\io} \dd t\, \ee^{\ii \o t} \ka_{mn}(t)$,
noting that $\ka_{mn}(t) = 0$ for $t < 0$ (see Appendix~\ref{App:Linear-response_theory:Linear_response_in_closed_quantum_systems}), and changing notation from $\langle \cdots \rangle_{\b}$ to $\langle \cdots \rangle_{0}$ defined in \eqref{H_sCFT}.
\end{proof}

For the special case of a homogeneous system,
$\bigl\langle j_{m}(x; t) j_{n}(x'; \ii\t) \bigr\rangle_{\b}^{c, \io}$
depends only on $x-x'$ due to translation invariance.
Thus, by a change of variable, \eqref{ka_mn_o_GK_appendix} becomes
\begin{equation}
\label{ka_mn_t_homog}
\ka_{mn}(t)
= \frac{1}{\b} \int_{0}^{\b} \dd \t
	\int_{-\io}^{\io} \dd x\,
	\bigl\langle j_{m}(x; t) j_{n}(0; \ii\t) \bigr\rangle_{\b}^{c, \io},
\end{equation}
again with $\langle \cdots \rangle_{\b}$ defined below \eqref{R_mn_t}, since $\int_{-\io}^{\io} \dd x'\, \partial_{x'} [-W(x')] = W(-\io) - W(\io) = 1$.


\subsection{Alternative linear-response functions}
\label{App:Linear-response_theory:Alternative_linear-response_functions}


A property of the response functions defined in \eqref{R_mn_t_s} is that they will depend on the details of the perturbations if the system is inhomogeneous.
One way to circumvent this is to define alternative linear-response functions as follows.

In analogy with Appendix~\ref{App:Linear-response_theory:Linear_response_in_closed_quantum_systems},  introduce
$H_{\mathrm{sys}}(\boldsymbol{\la}(\cdot, s))
= H_{\mathrm{sys}} - \sum_{n \geq 1} \int_{-L/2}^{L/2} \dd x\, \la_{n}(x, s) \rho_{n}(x)$,
$\boldsymbol{\la}(x, s) = (\la_{1}(x, s), \la_{2}(x, s), \ldots)$, where $\la_{n}(x, s)$
are functions of position $x$ and time $s$ satisfying $\la_{n}(x + L, s) = \la_{n}(x, s)$ and $\rho_{n}(x)$ are densities corresponding to the conserved quantities $Q_{n} = \int \dd x\, \rho_{n}(x)$ ($n = 1, 2, \ldots$).
Instead of \eqref{rho_t_lambda}, consider the time-evolved state
\begin{equation}
\label{rho_t_lambda_x}
\hat{\rho}(t, \boldsymbol{\la}(\cdot, \cdot))
= \overleftarrow{\cT} \!
	\exp \biggl(
		-\ii \int_{t_{0}}^{t} \dd s\, H_{\mathrm{sys}}(\boldsymbol{\la}(\cdot, s))
	\biggr)
	\hat{\rho} \,
	\overrightarrow{\cT} \!
	\exp \biggl(
		\ii \int_{t_{0}}^{t} \dd s\, H_{\mathrm{sys}}(\boldsymbol{\la}(\cdot, s))
	\biggr).
\end{equation}
For operators $\cO_{m}(x)$ ($m = 1, 2, \ldots$), define the response functions
\begin{equation}
\label{R_mn_x_t_xp_s}
R_{mn}(x, t; x', s)
=	\frac{\d}{\d \la_{n}(x', s)}
	\Bigl(
		\Tr \bigl[ \cO_{m}(x) \hat{\rho}(t, \boldsymbol{\la}(\cdot, \cdot)) \bigr]
	\Bigr) \Big|_{\boldsymbol{\la}(\cdot, \cdot) = {\bf 0}}
\end{equation}
for $s > t_{0}$.
Again, without loss of generality, set $t_{0} = -\io$.
Repeating the steps in Appendix~\ref{App:Linear-response_theory:Linear_response_in_closed_quantum_systems} yields $R_{mn}(x, t; x', s) = R_{mn}(x, t - s; x', 0)$ and, defining $R_{mn}(x, t; x') = R_{mn}(x, t; x', 0)$, we get
\begin{equation}
\label{R_mn_x_t_xp}
R_{mn}(x, t; x')
= \int_{0}^{\b} \dd \t\,
	\bigl\langle \cO_{m}(x; t) \partial\pdag_{\ii\t} \rho_{n}(x'; \ii\t) \bigr\rangle_{\b}
\end{equation}
with $\cO_{m}(x; t) = \ee^{\ii H_{\mathrm{sys}} t} \cO_{m}(x) \ee^{-\ii H_{\mathrm{sys}} t}$ and $\langle \cdots \rangle_{\b}$ defined below \eqref{R_mn_t}.

Similar to Appendix~\ref{App:Linear-response_theory:Conductivity_matrix}, in this case picking $\cO_{m}(x) = j_{m}(x)$, we can define the conductivities $\ka_{mn}(x, t; x')$ by
\begin{equation}
\label{conductivity_def_x_t_xp}
\partial_{x'} \ka_{mn}(x, t; x')
= \lim_{L\to\io} - \frac{1}{\beta} R_{mn}(x, t; x')
\end{equation}
in the thermodynamic limit.
Repeating the steps in the proof in Appendix~\ref{App:Linear-response_theory:Conductivity_matrix}, assuming that the current-current correlation functions decay rapidly for large separations, yields
\begin{equation}
\label{ka_mn_x_t_xp_GK}
\ka_{mn}(x, t; x')
= \frac{1}{\b} \int_{0}^{\b} \dd \t
	\langle j_{m}(x; t) j_{n}(x'; \ii\t) \rangle_{0}^{c, \io},
\end{equation}
from which \eqref{ka_mn_p_o_xp_GK} follows.

For the special case of a homogeneous system, we obtain
$\ka_{mn}(x, t; x') = \ka_{mn}(x - x', t; 0)$, which yields the result in \eqref{ka_mn_t_homog} if we define $\ka_{mn}(t) = \int_{-\io}^{\io} \dd x\, \ka_{mn}(x, t; 0)$.


\section{Computational details}
\label{App:Computational_details}


In this appendix, we give computational details for the results in Sect.~\ref{Sec:Applications}.

For later reference, we collect formulas for equilibrium expectations in the thermodynamic limit of the components of the energy-momentum tensor and the $\mathrm{U}(1)$ current:
\begin{subequations}
\label{Tpm_v0b0m0_expectations}
\begin{align}
\langle T_{\pm}(x) \rangle_{0}^{\io}
& = \frac{\pi c}{12 (v_{0}\b_{0})^2} + \frac{\ka\m_{0}^2}{4\pi v_{0}^2},
		\label{Tpm_v0b0m0_expectation} \\
\langle T_{\pm}(x_{1})T_{\pm}(x_{2}) \rangle_{0}^{\io}
& = \biggl( \frac{\pi c}{12 (v_{0}\b_{0})^2} + \frac{\ka\m_{0}^2}{4\pi v_{0}^2} \biggr)^2
		+ \frac{\pi^2 c}{8(v_{0}\b_{0})^4 \sinh^4 \bigl( \pi[x_{1}-x_{2}]/v_{0}\b_{0} \bigr)}
		\nonumber \\
& \quad	
		+ \frac{-\ka\m_{0}^2/v_{0}^2}
				{4(v_{0}\b_{0})^2 \sinh^2 \bigl( \pi[x_{1}-x_{2}]/v_{0}\b_{0} \bigr)},
		\label{Tpm_Tpm_v0b0m0_expectation} \\
\langle T_{\pm}(x_{1})T_{\mp}(x_{2}) \rangle_{0}^{\io}
& = \biggl( \frac{\pi c}{12 (v_{0}\b_{0})^2} + \frac{\ka\m_{0}^2}{4\pi v_{0}^2} \biggr)^2
		\label{Tpm_Tmp_v0b0m0_expectation}
\end{align}
\end{subequations}
and
\begin{subequations}
\label{Jpm_v0b0m0_expectations}
\begin{align}
\langle J_{\pm}(x) \rangle_{0}^{\io}
& = \frac{\ka\m_{0}}{2\pi v_{0}},
		\label{Jpm_v0b0m0_expectation} \\
\langle J_{\pm}(x_{1})J_{\pm}(x_{2}) \rangle_{0}^{\io}
& = \biggl( \frac{\ka\m_{0}}{2\pi v_{0}} \biggr)^2
		+ \frac{-\ka}{4(v_{0}\b_{0})^2 \sinh^2 \bigl( \pi[x_{1}-x_{2}]/v_{0}\b_{0} \bigr)},
		\label{Jpm_Jpm_v0b0m0_expectation} \\
\langle J_{\pm}(x_{1})J_{\mp}(x_{2}) \rangle_{0}^{\io}
& = \biggl( \frac{\ka\m_{0}}{2\pi v_{0}} \biggr)^2
		\label{Jpm_Jmp_v0b0m0_expectation}
\end{align}
\end{subequations}
as well as
\begin{subequations}
\label{Tpm_Jpm_v0b0m0_expectations}
\begin{align}
\langle T_{\pm}(x_{1})J_{\pm}(x_{2}) \rangle_{0}^{\io}
& = \langle J_{\pm}(x_{1})T_{\pm}(x_{2}) \rangle_{0}^{\io} \nonumber \\
& = \biggl( \frac{\pi c}{12 (v_{0}\b_{0})^2} + \frac{\ka\m_{0}^2}{4\pi v_{0}^2} \biggr) \frac{\ka\m_{0}}{2\pi v_{0}}
		+ \frac{-\ka\m_{0}/v_{0}}{4(v_{0}\b_{0})^2 \sinh^2 \bigl( \pi[x_{1}-x_{2}]/v_{0}\b_{0} \bigr)},
		\label{Tpm_Jpm_v0b0m0_expectation} \\
\langle T_{\pm}(x_{1})J_{\mp}(x_{2}) \rangle_{0}^{\io}
& = \langle J_{\pm}(x_{1})T_{\mp}(x_{2}) \rangle_{0}^{\io}  \nonumber \\
& = \biggl( \frac{\pi c}{12 (v_{0}\b_{0})^2} + \frac{\ka\m_{0}^2}{4\pi v_{0}^2} \biggr) \frac{\ka\m_{0}}{2\pi v_{0}},
		\label{Tpm_Jmp_v0b0m0_expectation}
\end{align}
\end{subequations}
cf., e.g., Sections~3.3,~4.1, and~4.3 in \cite{GLM}.
We note that translation invariance is manifest and that \eqref{Tpm_v0b0m0_expectations}\textnormal{--}\eqref{Tpm_Jpm_v0b0m0_expectations} only hold true in the thermodynamic limit.

Following \cite{GLM}, the formulas in \eqref{Tpm_v0b0m0_expectations}\textnormal{--}\eqref{Tpm_Jpm_v0b0m0_expectations} for $\m_{0} = 0$ can be obtained in a simple way for any modular-invariant CFT.
In this case, the reason why these hold true only in the limit $L\to\io$ is that then the only contributions are from the vacuum expectation in the dual representation on the circle with circumference $v_{0}\b_{0}$, which is universal since it depends only on the two vacuum highest-weight representations of the Virasoro algebra.
If we did not take $L\to\io$, there would be contributions from the other Verma modules that depend on the representation content of the CFT, i.e., all eigenstates $|h_+,h_-\rangle$ of $L^{\pm}_0$ where $L^{\pm}_0|h_+,h_-\rangle = h_{\pm}|h_+,h_-\rangle$ and not only the vacuum $|0\rangle$ corresponding to $h_+ = h_- = 0$, see, e.g., \cite{FMS}.
Lastly, we mention that the formulas for $\m_{0} \neq 0$ can be obtained by large gauge transformations.%
\footnote{%
Strictly speaking, the possible values of $\m_{0}$ are constrained by that $L\m_{0} / 2\pi v_{0}$ must be an integer, but this is of no consequence in the infinite volume.%
}


\subsection{Proof of \texorpdfstring{\eqref{cEcJ_rj_iCFT_neq}}{}}
\label{App:Computational_details:cEcJ_rj_iCFT_neq}


It follows from Proposition~\ref{Proposition:Main} that
\begin{align}
\langle T_{\pm}(x; t) \rangle_{\mathrm{neq}}
& = 	\biggl( \frac{v_{0}\b_{0}}{v(x)\b(\tilde{x}^{\mp})} \biggr)^2
			\langle T_{\pm}(g(\tilde{x}^{\mp})) \rangle_{0}
		+ \frac{v_{0}\b_{0}
					[\m(\tilde{x}^{\mp})\b(\tilde{x}^{\mp}) - \m_{0}\b_{0}]}
				{[v(x)\b(\tilde{x}^{\mp})]^2}
			\langle  J_{\pm}(g(\tilde{x}^{\mp})) \rangle_{0} \nonumber \\
& \quad
		+ \frac{\ka}{4\pi}
			\biggl(
				\frac{\m(\tilde{x}^{\mp})\b(\tilde{x}^{\mp}) - \m_{0}\b_{0}}
					{v(x)\b(\tilde{x}^{\mp})} \biggr)^2
		- \frac{\cT(\tilde{x}^{\mp}) + \cS(x)}{2 v(x)^2},
\end{align}
which, using \eqref{Tpm_v0b0m0_expectation} and \eqref{Jpm_v0b0m0_expectation}, implies
\begin{equation}
\langle T_{\pm}(x; t) \rangle_{\mathrm{neq}}^{\io}
= \frac{\pi c}{12 [v(x)\b(\tilde{x}^{\mp})]^2}
	+ \frac{\ka \m(\tilde{x}^{\mp})^2}{4\pi v(x)^2}
	- \frac{\cT(\tilde{x}^{\mp}) + \cS(x)}{2 v(x)^2}
\end{equation}
in the thermodynamic limit.
This together with \eqref{H_iCFT} and \eqref{cJ_iCFT} yields \eqref{cEcJ_iCFT_neq} with $F(x)$ in \eqref{F_x_G_x_iCFT_neq}.
Similarly, it follows from Proposition~\ref{Proposition:Main} that
\begin{equation}
\langle J_{\pm}(x; t) \rangle_{\mathrm{neq}}
= \frac{v_{0}\b_{0}}{v(x)\b(\tilde{x}^{\mp})}
	\langle J_{\pm}(g(\tilde{x}^{\mp})) \rangle_{0}
	+ \frac{\ka}{2\pi} \frac{\m(\tilde{x}^{\mp})\b(\tilde{x}^{\mp}) - \m_{0}\b_{0}}{v(x)\b(\tilde{x}^{\mp})},
\end{equation}
which, using \eqref{Jpm_v0b0m0_expectation}, implies
\begin{equation}
\langle J_{\pm}(x; t) \rangle_{\mathrm{neq}}^{\io}
= \frac{\ka \m(\tilde{x}^{\mp})}{2\pi v(x)}.
\end{equation}
This together with \eqref{Q_iCFT} and \eqref{j_iCFT} yields \eqref{rj_iCFT_neq} with $G(x)$ in \eqref{F_x_G_x_iCFT_neq}.


\subsection{Proofs of \texorpdfstring{\eqref{cJ_j_combos_iCFT_neq}}{} and \texorpdfstring{\eqref{psipsi_iLLM_neq}}{}}
\label{App:Computational_details:cJ_j_combos_iCFT_neq_and_psipsi_iLLM_neq}


The results in \eqref{cJ_j_combos_iCFT_neq} follow by straightforward but tedious computations using Proposition~\ref{Proposition:Main}, \eqref{cJ_iCFT}, \eqref{j_iCFT}, and \eqref{Tpm_v0b0m0_expectations}\textnormal{--}\eqref{Tpm_Jpm_v0b0m0_expectations}.
Similarly, the result in \eqref{psipsi_iLLM_neq} also follows straightforwardly from Proposition~\ref{Proposition:Main} and the equilibrium expectation for the renormalized fermionic fields in the usual (homogeneous) local Luttinger model:%
\begin{multline}
\langle \psi^+_{r}(x^-_{1},x^+_{1}) \psi^-_{r'}(x^-_{2},x^+_{2}) \rangle_{0}^{\io}
= \d_{r, r'}
	\frac{1}{2\pi\tilde{\ell}} \,
	\ee^{
		  \ii \m_{0}[
		  		\tau^+_{\psi^-_r} (x^-_{1} - x^-_{2})
				- \tau^-_{\psi^-_r} (x^+_{1} - x^+_{2})
			]/v_{0}
	} \\
\times
	\left(
		\frac{
			\ii \pi\tilde{\ell}
		}{
			v_{0}\b_{0}
			\sinh( \pi[x^-_{1} - x^-_{2}]/v_{0}\b_{0} )
		}
	\right)^{2\D^+_{\psi_r}}
	\left( 
		\frac{
			-\ii \pi\tilde{\ell}
		}{
			v_{0}\b_{0}
			\sinh( \pi[x^+_{1} - x^+_{2}]/v_{0}\b_{0} )
		}
	\right)^{2\D^-_{\psi_r}},
\end{multline}
where $x_{1,2}^{\pm} = x_{1,2} \pm v_{0} t_{1,2}$ and $\tilde{\ell}$ is a length parameter introduced in the multiplicative renormalization of the fields, cf., e.g., \cite{LLMM2, LaMo1}.
As is commonplace, we set $\tilde{\ell} = 1$.


\subsection{Proof of \texorpdfstring{\eqref{D_mn_ka_mn_o_reg_iCFT}}{} starting from \texorpdfstring{\eqref{ka_mn_o_dynamic}}{}}
\label{App:Computational_details:ka_mn_o_dynamic_iCFT}


It follows from \eqref{cEcJ_rj_iCFT_neq} that
$\langle j_{1}(x; t) \rangle_{\mathrm{neq}}^{\io}
= \sum_{r=\pm} r G(\tilde{x}^{-r})/2$
and
$\langle j_{2}(x; t) \rangle_{\mathrm{neq}}^{\io}
= \sum_{r=\pm} r F(\tilde{x}^{-r})/2$
with
\begin{equation}
G(x)
= - \frac{\ka \m_{1}(x)}{\pi \m_{2}(x)},
\quad
F(x)
= \frac{\pi^2 c + 3\ka \m_{1}(x)^2}{6\pi \m_{2}(x)^2}
	+ \frac{c v(x)^2}{12\pi}
		\Biggl[
			\frac{\m_{2}''(x)}{\m_{2}(x)}
			- \frac{1}{2} \biggl( \frac{\m_{2}'(x)}{\m_{2}(x)} \biggr)^2
			+ \frac{v'(x)}{v(x)} \frac{\m_{2}'(x)}{\m_{2}(x)}
		\Biggr].
\end{equation}
Since
\begin{subequations}
\begin{align}
& \frac{\partial}{\partial(\d\m_{1})} G(x) \bigg|_{\boldsymbol{\d\m} = 0}
	= - \frac{\ka}{\pi \m_{2}} W(x), \\
& \frac{\partial}{\partial(\d\m_{2})} G(x) \bigg|_{\boldsymbol{\d\m} = 0}
	= \frac{\partial}{\partial(\d\m_{1})} F(x) \bigg|_{\boldsymbol{\d\m} = 0}
	= \frac{\ka \m_{1}}{\pi \m_{2}^2} W(x), \\
& \frac{\partial}{\partial(\d\m_{2})} F(x) \bigg|_{\boldsymbol{\d\m} = 0}
	= - \frac{\pi^2 c + 3\ka \m_{1}^2}{3\pi \m_{2}^3} W(x)
		+ \frac{c}{12\pi \m_{2}} v(x) \partial_{x} \bigl[ v(x) \partial_{x} W(x) \bigr]
\end{align}
\end{subequations}
and since \eqref{dxpm_dx} implies
$\partial_{t} u(\tilde{x}^{-r})
= -r v(x) \partial_{x} u(\tilde{x}^{-r})
= -r v(\tilde{x}^{-r}) \partial_{\tilde{x}^{-r}} u(\tilde{x}^{-r})$ for any differentiable function $u(x)$, we obtain
\begin{subequations}
\begin{align}
& \frac{\partial}{\partial(\d\m_{1})}
	\partial_{t} \langle j_{1}(x; t) \rangle_{\mathrm{neq}}^{\io}
	\bigg|_{\boldsymbol{\d\m} = 0}
	= \frac{\ka}{2\pi \m_{2}}
		\sum_{r=\pm}
		v(\tilde{x}^{-r}) \partial_{\tilde{x}^{-r}} W(\tilde{x}^{-r}), \\
& \frac{\partial}{\partial(\d\m_{2})}
	\partial_{t}\langle j_{1}(x; t) \rangle_{\mathrm{neq}}^{\io}
	\bigg|_{\boldsymbol{\d\m} = 0}
	= \frac{\partial}{\partial(\d\m_{1})}
		\partial_{t} \langle j_{2}(x; t) \rangle_{\mathrm{neq}}^{\io}
		\bigg|_{\boldsymbol{\d\m} = 0} 
	= - \frac{\ka \m_{1}}{2\pi \m_{2}^2}
		\sum_{r=\pm}
		v(\tilde{x}^{-r}) \partial_{\tilde{x}^{-r}} W(\tilde{x}^{-r}), \\
& \frac{\partial}{\partial(\d\m_{2})}
	\partial_{t} \langle j_{2}(x; t) \rangle_{\mathrm{neq}}^{\io}
	\bigg|_{\boldsymbol{\d\m} = 0}
	= \frac{\ka \m_{1}^2}{2\pi \m_{2}^3}
		\sum_{r=\pm} 
		v(\tilde{x}^{-r}) \partial_{\tilde{x}^{-r}}
		W(\tilde{x}^{-r}) \nonumber \\
& \qquad \qquad \qquad \;\;
	+ \frac{\pi c}{6 \m_{2}^3}
		\sum_{r=\pm} 
		v(\tilde{x}^{-r}) \partial_{\tilde{x}^{-r}}
		\biggl[
			W(\tilde{x}^{-r})
			- \Bigl( \frac{\m_{2}}{2\pi} \Bigr)^2
				v(\tilde{x}^{-r}) \partial_{\tilde{x}^{-r}}
				\bigl[
					v(\tilde{x}^{-r}) \partial_{\tilde{x}^{-r}} W(\tilde{x}^{-r})
				\bigr]
		\biggr].
\end{align}
\end{subequations}
Inserting these into \eqref{ka_mn_o_dynamic} followed by a change of variable to
$y = f(\tilde{x}^{-r})$ with $v_{0}$ replaced by $v$ and using
$\m_{1} = \b\m$ and $\m_{2} = -\b$
yields
\begin{equation}
\label{ka_mn_o_I_1_I_2}
\ka_{11}(\o)
	= \frac{\ka}{2\pi \b} I_{1}(\o),
\quad
\ka_{12}(\o)
	= \ka_{21}(\o)
	= \frac{\ka \m}{2\pi \b} I_{1}(\o),
\quad
\ka_{22}(\o)
	= \frac{\pi c}{6 \b^3} I_{2}(\o) + \frac{\ka \m^2}{2\pi \b} I_{1}(\o),
\end{equation}
where
\begin{subequations}
\label{I_1_I_2_pos_space}
\begin{align}
I_{1}(\o)
& = \sum_{r=\pm}
		\int_{0}^{\io} \dd t\, \ee^{\ii\o t}
		\int_{-\io}^{\io} \dd y\,
		v( f^{-1}(y + rvt) )
		\partial_{y} \bigl[ - W(f^{-1}(y)) \bigr], \\
I_{2}(\o)
& = \sum_{r=\pm}
		\int_{0}^{\io} \dd t\, \ee^{\ii\o t}
		\int_{-\io}^{\io} \dd y\,
		v( f^{-1}(y + rvt) )
		\left[ 1 - \left( \frac{v\b \partial_{y}}{2\pi} \right)^2 \right]
		\partial_{y} \left[ - W(f^{-1}(y)) \right].
\end{align}
\end{subequations}
One can show that the latter can be rewritten using
\begin{equation}
\label{kTh}
k(y)
= \int_{-\io}^{\io} \dd y'\,
	\frac{v( f^{-1}(y+y') )}{v}
	\partial_{y'} \bigl[ -W(f^{-1}(y')) \bigr],
\quad
\hat{k}(p) = \int_{-\io}^{\io} \dd y\, k(y) \ee^{-\ii py}
\end{equation}
as follows
\begin{equation}
\label{I_1_I_2_mom_space}
I_{1}(\o)
= \sum_{r=\pm}
	\int_{-\io}^{\io} \frac{\dd p}{2\pi}
	\frac{\ii v \hat{k}(p)}{\o + rvp + \ii0^+},
\quad
I_{2}(\o)
= \sum_{r=\pm}
	\int_{-\io}^{\io} \frac{\dd p}{2\pi}
	\left[ 1 + \left( \frac{v\b p}{2\pi} \right)^2 \right]
	\frac{\ii v \hat{k}(p)}{\o + rvp + \ii0^+}.
\end{equation}
To see this, note that one can write $\hat{k}(p) = \hat{k}_{1}(p) \hat{k}_{2}(-p)$ with $\hat{k}_{1}(p)$ and $\hat{k}_{2}(p)$ corresponding to $k_{1}(y) = {v( f^{-1}(y) )}/{v}$ and $k_{2}(y) = \partial_{y} \bigl[ -W(f^{-1}(y)) \bigr]$, respectively, which means that, e.g.,
$\int \dd y\, k_{1}(y + rvt) [1 - (v\b \partial_{y}/2\pi)^2] k_{2}(y)
= (2\pi)^{-1} \int \dd p\, \hat{k}_{1}(p) \hat{k}_{2}(-p) [1 + (v\b p/2\pi)^2] \ee^{\ii rvp t}$,
and use
\begin{equation}
\label{time_integral}
\int_{0}^{\io} \dd t\, \ee^{\ii(\o + rvp)t} = \frac{\ii}{\o + rvp + \ii0^+}.
\end{equation}
In Appendix~\ref{Appendix:Computational_details:I_1_I_2_results}, we show that
\begin{equation}
\label{I_1_I_2_results}
\Re I_{1}(\o) = 2\pi v \d(\o) + I(\o),
\quad
\Re I_{2}(\o) = 2\pi v \d(\o) + \left[ 1 + \left( \frac{\o\b}{2\pi} \right)^2 \right] I(\o)
\end{equation}
with $I(\o)$ in \eqref{I_omega}, which together with \eqref{ka_mn_o_I_1_I_2} completes the proof.


\subsection{Proof of \texorpdfstring{\eqref{D_mn_ka_mn_o_reg_iCFT}}{} starting from \texorpdfstring{\eqref{ka_mn_o_GK}}{}}
\label{App:Computational_details:ka_mn_o_GK_iCFT}


It suffices to derive \eqref{ka_mn_o_I_1_I_2} together with \eqref{I_1_I_2_mom_space}.
To do so, we will need the following integrals:%
\begin{subequations}
\label{F_a_b}
\begin{align}
\int_{-\io}^{\io} \dd \xi\, \frac{\ee^{\ii b \xi}}{\sinh^4(\xi + ia)}
& = \frac{\pi(b^3+4b)}{3} \frac{\ee^{b[a]_{\pi}}}{\ee^{b\pi}-1},
		\label{F_th_a_b} \\
\int_{-\io}^{\io} \dd \xi\, \frac{\ee^{\ii b \xi}}{\sinh^2(\xi + ia)}
& = -2\pi b \frac{\ee^{b[a]_{\pi}}}{\ee^{b\pi}-1}
		\label{F_el_a_b}
\end{align}
\end{subequations}
for all $a, b \in \mathbb{R}$, where $[a]_{\pi} \in [0, \pi)$ is defined by
$a = n_0 \pi + [a]_{\pi}$ for $n_0 \in \mathbb{Z}$.
(These can be proven using the residue theorem.)
We will also need the equilibrium current-current correlation functions.
Setting $\b(x) = \b$ and $\m(x) = \m$ in \eqref{cJ_j_combos_iCFT_neq}, we obtain
\begin{subequations}
\label{cJ_j_combos_iCFT_eq}
\begin{align}
\langle j_{2}(x_{1}; t_{1}) j_{2}(x_{2}; t_{2}) \rangle_{0}^{c, \io}
& = \sum_{r=\pm}
		\Biggl[
			\frac{\pi^2 c}
				{8\b^4
					\sinh^4
					\bigl(
						\pi[f(x_{1}) - f(x_{2}) - r v (t_{1} - t_{2})]/v\b
					\bigr)} \nonumber \\
& \qquad\quad +
			\frac{-\ka\m^2}
					{4\b^2
						\sinh^2
						\bigl(
							\pi[f(x_{1}) - f(x_{2}) - r v (t_{1} - t_{2})]/v\b
						\bigr)}
		\Biggr],
		\label{cJcJ_iCFT_eq} \\
\langle j_{1}(x_{1}; t_{1}) j_{1}(x_{2}; t_{2}) \rangle_{0}^{c, \io}
& = \sum_{r=\pm}
		\frac{-\ka}
			{4\b^2
				\sinh^2
				\bigl(
					\pi[f(x_{1}) - f(x_{2}) - r v (t_{1} - t_{2})]/v\b
				\bigr)},
		\label{jj_iCFT_eq} \\
\langle j_{2}(x_{1}; t_{1}) j_{1}(x_{2}; t_{2}) \rangle_{0}^{c, \io}
& = \langle j_{1}(x_{1}; t_{1}) j_{2}(x_{2}; t_{2}) \rangle_{0}^{c, \io}
	= \m \langle j(x_{1}; t_{1}) j(x_{2}; t_{2}) \rangle_{0}^{c, \io}
		\label{jcJ_cJe_iCFT_eq}
\end{align}
\end{subequations}
for $f(x)$ with $v_{0}$ replaced by $v$, where we used that $f(\tilde{x}^{-r}_{j}) = f(x_{j}) - r v t_{j}$.

Consider first $\ka_{22}(\o)$.
By inserting \eqref{cJcJ_iCFT_eq} into \eqref{ka_mn_o_GK}, changing variables to $y = f(x)$ and $y' = f(x')$, and shifting $y$ to $y + y'$, we obtain
\begin{align}
\ka_{22}(\o)
& = \frac{\pi^2 c}{8\b^5}
		\sum_{r=\pm}
		\int_{0}^{\b} \dd \t
		\int_{0}^{\io} \dd t\, \ee^{\ii\o t}
		\int_{-\io}^{\io} \dd y\,
		\frac{k(y)}{\sinh^4 \Bigl( \pi [ y - rv(t-\ii\t) ]/v\b \Bigr)} \nonumber \\
& \quad
	- \frac{\ka \m^2}{4\b^3}
		\sum_{r=\pm}
		\int_{0}^{\b} \dd \t
		\int_{0}^{\io} \dd t\, \ee^{\ii\o t}
		\int_{-\io}^{\io} \dd y\,
		\frac{k(y)}{\sinh^2 \Bigl( \pi [ y - rv(t-\ii\t) ]/v\b \Bigr)}
	\label{ka_22_GK_step_1}
\end{align}
with $k(y)$ in \eqref{kTh}.
Inserting
$k(y) = (2\pi)^{-1} \int_{-\io}^{\io} \dd p\, \hat{k}(p) \ee^{\ii py}$
into \eqref{ka_22_GK_step_1} followed by a change of variable to $\xi = \pi(y - rvt)/v\b$ gives
\begin{align}
\ka_{22}(\o)
& = \frac{\pi vc}{8\b^4}
		\sum_{r=\pm}
		\int_{-\io}^{\io} \frac{\dd p}{2\pi}
		\hat{k}(p)
		\int_{0}^{\b} \dd \t \int_{0}^{\io} \dd t\, \ee^{\ii(\o + rvp)t}
		\int_{-\io}^{\io} \dd \xi\,
		\frac{\ee^{\ii pv\b\xi/\pi}}{\sinh^4( \xi + \ii r\pi\t/\b)} \nonumber \\
& \quad
	- \frac{v \ka \mu^2}{4\pi\b^2}
		\sum_{r=\pm}
		\int_{-\io}^{\io} \frac{\dd p}{2\pi}
		\hat{k}(p)
		\int_{0}^{\b} \dd \t \int_{0}^{\io} \dd t\, \ee^{\ii(\o + rvp)t}
		\int_{-\io}^{\io} \dd \xi\,
		\frac{\ee^{\ii pv\b\xi/\pi}}{\sinh^2( \xi + \ii r\pi\t/\b)}.
	\label{ka_22_GK_step_2}
\end{align}
The $\xi$-integrals are of the form in \eqref{F_a_b}.
Using these formulas with $a = r\pi \t/\b$ and $b = v\b p/\pi$, computing the $\t$-integral by treating the cases $r = \pm$ separately, and computing the $t$-integral using \eqref{time_integral}, it follows that \eqref{ka_22_GK_step_2} yields $\ka_{22}(\o)$ in \eqref{ka_mn_o_I_1_I_2} with $I_{1}(\o)$ and $I_{2}(\o)$ in \eqref{I_1_I_2_mom_space}.
The corresponding results for the remaining $\ka_{mn}(\o)$ follow analogously.


\subsection{Proof of \texorpdfstring{\eqref{I_1_I_2_results}}{}}
\label{Appendix:Computational_details:I_1_I_2_results}


To prove \eqref{I_1_I_2_results}, we first note that for $v(x) = v$, \eqref{kTh} implies $k(y) = W(-\io) - W(\io) = 1$, meaning that $\hat{k}(p) = 2\pi \d(p)$.
Thus, in this case, \eqref{I_1_I_2_mom_space} gives
\begin{equation}
I_{1}(\o)
= I_{2}(\o) 
= \frac{2 \ii v}{\o + \ii0^+}
= 2v \bigl[ \pi \d(\o)  + \ii \cP(1/\o) \bigr],
\end{equation}
where $\cP$ denotes the principal value.
This corresponds to the Drude peaks which are the sole contributions to the conductivities in standard CFT.

For inhomogeneous CFT, it thus remains to consider the remainders obtained by subtracting the above from \eqref{I_1_I_2_mom_space}.
I.e.,
\begin{subequations}
\label{Delta_I_1_I_2}
\begin{align}
\D I_{1}(\o)
& =	I_{1}(\o) - \frac{2 \ii v}{\o + \ii0^+}
	= \sum_{r=\pm}
		\int_{-\io}^{\io} \frac{\dd p}{2\pi}
		\frac{\ii v\D \hat{k}(p)}{\o + rvp + \ii0^+}, \\
\D I_{2}(\o)
& =	I_{2}(\o) - \frac{2 \ii v}{\o + \ii0^+}
	= \sum_{r=\pm}
		\int_{-\io}^{\io} \frac{\dd p}{2\pi}
		\left[ 1 + \left( \frac{v\b p}{2\pi} \right)^2 \right]
		\frac{\ii v\D \hat{k}(p)}{\o + rvp + \ii0^+},
\end{align}
\end{subequations}
where
\begin{equation}
\label{Delta_kTh}
\D \hat{k}(p)
= \hat{k}(p) - 2\pi \d(p)
= \int_{-\io}^{\io} \dd y \int_{-\io}^{\io} \dd y'\, K(y, y') \ee^{-\ii p(y-y')}
\end{equation}
with
$K(y, y')
= [v( f^{-1}(y) )/v] \partial_{y'} \bigl[ -W(f^{-1}(y')) \bigr]
	- \partial_{y'} \bigl[ -W(y') \bigr]$ for $f(x)$ with $v_{0}$ replaced by $v$.
Changing order of the integrals in \eqref{Delta_I_1_I_2} and \eqref{Delta_kTh}, which is possible by Fubini's theorem if $1 - v/v(x) \in L^{1}(\mathbb{R})$, and using
\begin{subequations}
\begin{align}
\sum_{r=\pm}
\int_{-\io}^{\io} \frac{\dd p}{2\pi}
\frac{\ii r/v}{p + r(\o + \ii0^+)/v}
\ee^{-\ii p(y-y')}
& = \frac{1}{v}
		\ee^{\ii\o|y-y'|/v}, \\
\sum_{r=\pm}
\int_{-\io}^{\io} \frac{\dd p}{2\pi}
\left[ 1 + \left( \frac{v\b p}{2\pi} \right)^2 \right]
\frac{\ii r/v}{p + r(\o + \ii0^+)/v}
\ee^{-\ii p(y-y')}
& = \frac{1}{v}
		\left[ 1 + \left( \frac{\o\b}{2\pi} \right)^2 \right]
		\ee^{\ii\o|y-y'|/v},
\end{align}
\end{subequations}
which follow from the residue theorem by dividing into the cases $y - y' > 0$ and $< 0$, we obtain%
\begin{subequations}
\begin{align}
I_{1}(\o)
& = \frac{2 \ii v}{\o + \ii0^+}
		+ \int_{-\io}^{\io} \dd y \int_{-\io}^{\io} \dd y'\, K(y, y')
			\ee^{\ii\o|y-y'|/v}, \\
I_{2}(\o)
& = \frac{2 \ii v}{\o + \ii0^+}
		+ \left[ 1 + \left( \frac{\o\b}{2\pi} \right)^2 \right]
			\int_{-\io}^{\io} \dd y \int_{-\io}^{\io} \dd y'\, K(y, y')
			\ee^{\ii\o|y-y'|/v}.
\end{align}
\end{subequations}
From this, \eqref{I_1_I_2_results} follows if we show that
\begin{equation}
\label{I_omega_identification}
I(\o)
= \Re \int_{-\io}^{\io} \dd y \int_{-\io}^{\io} \dd y'\, K(y, y') \ee^{\ii\o|y-y'|/v}.
\end{equation}
To prove the latter, change variable from $y$ to $s = (y-y')/v$ on the r.h.s., which gives
\begin{equation}
\label{yyp_integral}
\int_{-\io}^{\io} \dd y \int_{-\io}^{\io} \dd y'\, K(y, y')
	\cos \biggl( \frac{\o}{v}(y-y') \biggr)
= v \int_{-\io}^{\io} \dd s\, M(s) \cos(\o s)
\end{equation}
with
\begin{equation}
\label{I_integral}
M(s)
= \int_{-\io}^{\io} \dd y'\, K(vs + y', y')
= \int_{-\io}^{\io} \dd x'\,
		\biggl(
			\frac{v( f^{-1}(vs+f(x')) )}{v} - 1
		\biggr)
		\partial_{x'} \bigl[ -W(x') \bigr],
\end{equation}
where in the last step we changed variable to $x' = f^{-1}(y')$ in the first term and relabeled $y'$ by $x'$ in the second.
Inserting \eqref{I_integral} into \eqref{yyp_integral} and letting $x = f^{-1}(vs + f(x'))$ gives
\begin{multline}
\int_{-\io}^{\io} \dd y \int_{-\io}^{\io} \dd y'\, K(y, y')
	\cos \Bigl( \frac{\o}{v}(y-y') \Bigr) \\
= \int_{-\io}^{\io} \dd x \int_{-\io}^{\io} \dd x'\,
	\left( 1 - \frac{v}{v(x)} \right)
	\partial_{x'} \bigl[ -W(x') \bigr]
	\cos \Bigl( \frac{\o}{v} [f(x)-f(x')] \Bigr),
\end{multline}
which inserted into the r.h.s.\ of \eqref{I_omega_identification} gives the l.h.s.\ since $f(x) - f(x') = \int_{x'}^{x} \dd x''\, v/v(x'')$.


\end{appendices}


\let\oldbibliography\thebibliography
\renewcommand\thebibliography[1]{
  \oldbibliography{#1}
  \setlength{\parskip}{0pt}
  \setlength{\itemsep}{0pt + 0.85ex}
}




\begin{thebibliography}{}

\bibitem{Cardy:1988}
J.~L.~Cardy,
Conformal invariance and statistical mechanics,
in {\it Fields, Strings and Critical Phenomena}, Lecture Notes of Les Houches Summer School 1988, Session XLIX, edited by E.\ Br{\'e}zin and J.\ Zinn-Justin
(North-Holland, 1990), p.\ 169.

\bibitem{BPZ}
A.~A.~Belavin, A.~M.~Polyakov, and A.~B.~Zamolodchikov,
Infinite conformal symmetry in two-dimensional quantum field theory,
\href{https://doi.org/10.1016/0550-3213(84)90052-X}{Nucl.\ Phys.\ B {\bf 241}, 333 (1984)}.

\bibitem{CaCa1}
P.~Calabrese and J.~Cardy,
Entanglement and correlation functions following a local quench: a conformal field theory approach,
\href{https://doi.org/10.1088/1742-5468/2007/10/P10004}{J.\ Stat.\ Mech.\ (2007) P10004}.

\bibitem{BeDo1}
D.~Bernard and B.~Doyon,
Energy flow in non-equilibrium conformal field theory,
\href{https://doi.org/10.1088/1751-8113/45/36/362001}{J.\ Phys.\ A: Math.\ Theor.\ {\bf 45}, 362001 (2012)}.

\bibitem{BeDo2}
D.~Bernard and B.~Doyon,
Non-equilibrium steady states in conformal field theory,
\href{https://doi.org/10.1007/s00023-014-0314-8}{Ann.\ Henri Poincar{\'e} {\bf 16}, 113 (2015)}.

\bibitem{BDV}
D.~Bernard, B.~Doyon, and J.~Viti,
Non-equilibrium conformal field theories with impurities,
\href{https://doi.org/10.1088/1751-8113/48/5/05FT01}{J.\ Phys.\ A 48, 05FT01 (2015)}.

\bibitem{GaTa}
K.~Gaw\k{e}dzki and C.~Tauber,
Nonequilibrium transport through quantum-wire junctions and boundary defects for free massless bosonic fields,
\href{https://doi.org/10.1016/j.nuclphysb.2015.04.014}{Nucl.\ Phys.\ B {\bf 896}, 138 (2015)}.

\bibitem{CaCa2}
P.~Calabrese and J.~Cardy,
Quantum quenches in 1+1 dimensional conformal field theories,
\href{https://doi.org/10.1088/1742-5468/2016/06/064003}{J.\ Stat.\ Mech.\ (2016) 064003}.

\bibitem{LLMM1}
E.~Langmann, J.~L.~Lebowitz, V.~Mastropietro, and P.~Moosavi,
Steady states and universal conductance in a quenched Luttinger model,
\href{https://doi.org/10.1007/s00220-016-2631-x}{Commun.\ Math.\ Phys.\ {\bf 349}, 551 (2017)}.

\bibitem{LLMM2}
E.~Langmann, J.~L.~Lebowitz, V.~Mastropietro, and P.~Moosavi,
Time evolution of the Luttinger model with nonuniform temperature profile,
\href{https://doi.org/10.1103/PhysRevB.95.235142}{Phys.\ Rev.\ B {\bf 95}, 235142\ (2017)}.

\bibitem{GLM}
K.~Gaw\k{e}dzki, E.~Langmann, and P.~Moosavi,
Finite-time universality in nonequilibrium CFT,
\href{https://doi.org/10.1007/s10955-018-2025-x}{J.\ Stat.\ Phys.\ {\bf 172}, 353 (2018)}.

\bibitem{GaKo}
K.~Gaw\k{e}dzki and K.~K.~Koz{\l}owski,
Full counting statistics of energy transfers in inhomogeneous nonequilibrium states of (1+1)D CFT,
\href{https://doi.org/10.1007/s00220-020-03774-5}{Commun.\ Math.\ Phys.\ {\bf 377}, 1227 (2020)}.

\bibitem{SoCa}
S.~Sotiriadis and J.~Cardy,
Inhomogeneous quantum quenches,
\href{https://doi.org/10.1088/1742-5468/2008/11/P11003}{J.\ Stat.\ Mech.\ (2008) P11003}.

\bibitem{Kat2}
H.~Katsura,
Sine-square deformation of solvable spin chains and conformal field theories,
\href{https://doi.org/10.1088/1751-8113/45/11/115003}{J.\ Phys.\ A: Math.\ Theor.\ {\bf 45}, 115003 (2012)}.

\bibitem{WRL}
X.~Wen, S.~Ryu, and A.~W.~W.~Ludwig,
Evolution operators in conformal field theories and conformal mappings: entanglement Hamiltonian, the sine-square deformation, and others,
\href{https://doi.org/10.1103/PhysRevB.93.235119}{Phys.\ Rev.\ B {\bf 93}, 235119 (2016)}.

\bibitem{ADSV}
N.~Allegra, J.~Dubail, J.-M.~St{\'e}phan, and J.~Viti,
Inhomogeneous field theory inside the arctic circle,
\href{https://doi.org/10.1088/1742-5468/2016/05/053108}{J.\ Stat.\ Mech.\ (2016) 053108}.

\bibitem{DSVC}
J.~Dubail, J.-M.~St{\'e}phan, J.~Viti, and P.~Calabrese,
Conformal field theory for inhomogeneous one-dimensional quantum systems: the example of non-interacting Fermi gases,
\href{https://doi.org/10.21468/SciPostPhys.2.1.002}{SciPost Phys.\ {\bf 2}, 002 (2017)}.

\bibitem{DSC}
J.~Dubail, J.-M.~St{\'e}phan, and P.~Calabrese,
Emergence of curved light-cones in a class of inhomogeneous Luttinger liquids,
\href{https://doi.org/10.21468/SciPostPhys.3.3.019}{SciPost Phys.\ {\bf 3}, 019 (2017)}.

\bibitem{BrDu}
Y.~Brun and J.~Dubail,
The Inhomogeneous Gaussian Free Field, with application to ground state correlations of trapped 1d Bose gases,
\href{https://doi.org/10.21468/SciPostPhys.4.6.037}{SciPost Phys.\ {\bf 4}, 037 (2018)}.

\bibitem{WeWu1}
X.~Wen and J.-Q.~Wu,
Quantum dynamics in sine-square deformed conformal field theory: quench from uniform to non-uniform CFTs,
\href{https://doi.org/10.1103/PhysRevB.97.184309}{Phys.\ Rev.\ B {\bf 97}, 184309 (2018)}.

\bibitem{RBD}
P.~Ruggiero, Y.~Brun, and J.~Dubail,
Conformal field theory on top of a breathing one-dimensional gas of hard core bosons,
\href{https://doi.org/10.21468/SciPostPhys.6.4.051}{SciPost Phys.\ {\bf 6}, 051 (2019)}.

\bibitem{LaMo2}
E.~Langmann and P.~Moosavi,
Diffusive heat waves in random conformal field theory,
\href{https://doi.org/10.1103/PhysRevLett.122.020201}{Phys.\ Rev.\ Lett.\ {\bf 122}, 020201 (2019)}.

\bibitem{ABF}
V.~Alba, B.~Bertini, and M.~Fagotti,
Entanglement evolution and generalised hydrodynamics: interacting integrable systems,
\href{https://doi.org/10.21468/SciPostPhys.7.1.005}{SciPost Phys.\ {\bf 7}, 005 (2019)}.

\bibitem{BCRLM}
A.~Biella, M.~Collura, D.~Rossini, A.~De Luca, and L.~Mazza,
Ballistic transport and boundary resistances in inhomogeneous quantum spin chains,
\href{https://doi.org/10.1038/s41467-019-12784-4}{Nat.\ Commun.\ 10, 4820 (2019)}.

\bibitem{RCDD}
P.~Ruggiero, P.~Calabrese, B.~Doyon, and J.~Dubail,
Quantum generalized hydrodynamics,
\href{https://doi.org/10.1103/PhysRevLett.124.140603}{Phys.\ Rev.\ Lett.\ {\bf 124}, 140603 (2020)}.

\bibitem{Tom}
S.~Tomonaga,
Remarks on Bloch's method of sound waves applied to many-fermion problems,
\href{https://doi.org/10.1143/ptp/5.4.544}{Prog.\ Theor.\ Phys.\ {\bf 5}, 544 (1950)}.

\bibitem{Lut}
J.~M.~Luttinger,
An exactly soluble model of a many-fermion system,
\href{https://doi.org/10.1063/1.1704046}{J.\ Math.\ Phys.\ {\bf 4}, 1154 (1963)}.

\bibitem{MaLi}
D.~C.~Mattis and E.~H.~Lieb,
Exact solution of a many-fermion system and its associated boson field,
\href{https://doi.org/10.1063/1.1704281}{J.\ Math.\ Phys.\ {\bf 6}, 304 (1965)}.

\bibitem{Moo}
P.~Moosavi,
Non-equilibrium dynamics of exactly solvable quantum many-body systems,
\href{http://urn.kb.se/resolve?urn=urn%3Anbn%3Ase%3Akth%3Adiva-239155}{PhD thesis, KTH Royal Institute of Technology (2018)}.

\bibitem{AGV}
U.~Agrawal, S.~Gopalakrishnan, and R.~Vasseur,
Generalized hydrodynamics, quasiparticle diffusion, and anomalous local relaxation in random integrable spin chains,
\href{https://doi.org/10.1103/PhysRevB.99.174203}{Phys.\ Rev.\ B {\bf 99}, 174203 (2019)}.

\bibitem{CDY}
O.~A.~Castro-Alvaredo, B.~Doyon, and T.~Yoshimura,
Emergent hydrodynamics in integrable quantum systems out of equilibrium,
\href{https://doi.org/10.1103/PhysRevX.6.041065}{Phys.\ Rev.\ X {\bf 6}, 041065 (2016)}.

\bibitem{BCNF}
B.~Bertini, M.~Collura, J.~De~Nardis, and M.~Fagotti,
Transport in out-of-equilibrium $XXZ$ chains: exact profiles of charges and currents,
\href{https://doi.org/10.1103/PhysRevLett.117.207201}{Phys.\ Rev.\ Lett.\ {\bf 117}, 207201 (2016)}.

\bibitem{CaTo}
J.~Cardy and E.~Tonni,
Entanglement Hamiltonians in two-dimensional conformal field theory,
\href{https://doi.org/10.1088/1742-5468/2016/12/123103}{J.\ Stat.\ Mech.\ (2016) 123103}.

\bibitem{BeDo4}
D.~Bernard and B.~Doyon,
Diffusion and signatures of localization in stochastic conformal field theory,
\href{https://doi.org/10.1103/PhysRevLett.119.110201}{Phys.\ Rev.\ Lett.\ {\bf 119}, 110201 (2017)}.

\bibitem{BeDou}
D.~Bernard and P.~Le~Doussal,
Entanglement entropy growth in stochastic conformal field theory and the KPZ class,
\href{https://doi.org/10.1209/0295-5075/131/10007}{Europhys.\ Lett.\ {\bf 131}, 10007 (2020)}.

\bibitem{JoLi}
T.~De Jonckheere and J.~Lindgren,
Entanglement entropy in inhomogeneous quenches in $\text{AdS}_3/\text{CFT}_2$,
\href{https://doi.org/10.1103/PhysRevD.98.106006}{Phys.\ Rev.\ D {\bf 98}, 106006 (2018)}.

\bibitem{MLNR}
I.~MacCormack, A.~Liu, M.~Nozaki, and S.~Ryu,
Holographic duals of inhomogeneous systems: the rainbow chain and the sine-square deformation model,
\href{https://doi.org/10.1088/1751-8121/ab3944}{J.\ Phys.\ A: Math.\ Theor.\ {\bf 52}, 505401 (2019)}.

\bibitem{WeWu2}
X.~Wen and J.-Q.~Wu,
Floquet conformal field theory,
\href{https://arxiv.org/abs/1805.00031}{arXiv:1805.00031 [cond-mat.str-el] (2018)}.

\bibitem{LCTTNC1}
B.~Lapierre, K.~Choo, C.~Tauber, A.~Tiwari, T.~Neupert, and R.~Chitra,
Emergent black hole dynamics in critical Floquet systems,
\href{https://doi.org/10.1103/PhysRevResearch.2.023085}{Phys.\ Rev.\ Research {\bf 2}, 023085 (2020)}.

\bibitem{FGVW}
R.~Fan, Y.~Gu, A.~Vishwanath, and X.~Wen,
Emergent spatial structure and entanglement localization in Floquet conformal field theory,
\href{https://doi.org/10.1103/PhysRevX.10.031036}{Phys.\ Rev.\ X\ {\bf 10}, 031036 (2020)}.

\bibitem{HaWe}
B.~Han and X.~Wen,
Classification of $SL_{2}$ deformed Floquet conformal field theories,
\href{https://doi.org/10.1103/PhysRevB.102.205125}{Phys.\ Rev.\ B {\bf 102}, 205125 (2020)}.

\bibitem{LapMoo}
B.~Lapierre and P.~Moosavi,
Geometric approach to inhomogeneous Floquet systems,
\href{https://doi.org/10.1103/PhysRevB.103.224303}{Phys.\ Rev.\ B {\bf 103}, 224303 (2021)}.

\bibitem{Doy}
B.~Doyon,
Thermalization and pseudolocality in extended quantum systems,
\href{https://doi.org/10.1007/s00220-017-2836-7}{Commun.\ Math.\ Phys.\ {\bf 351}, 155 (2017)}.

\bibitem{Kubo}
R.~Kubo,
Statistical-mechanical theory of irreversible processes. I. General theory and simple applications to magnetic and conduction problems,
\href{https://doi.org/10.1143/JPSJ.12.570}{J.\ Phys.\ Soc.\ Jpn.\ {\bf 12}, 570 (1957)}.

\bibitem{Spo}
H.~Spohn,
Interacting and noninteracting integrable systems,
\href{https://doi.org/10.1063/1.5018624}{J.\ Math.\ Phys.\ {\bf 59}, 091402 (2018)}.

\bibitem{Schott}
M.~Schottenloher,
{\it A Mathematical Introduction to Conformal Field Theory}
\href{https://doi.org/10.1007/978-3-540-68628-6}{(Springer, Berlin, Heidelberg, 2008)}.

\bibitem{FMS}
P.~Di~Francesco, P.~Mathieu, and D.~S{\'e}n{\'e}chal,
{\it Conformal Field Theory}
\href{https://doi.org/10.1007/978-1-4612-2256-9}{(Springer, New York, 1997)}.

\bibitem{KnZa}
V.~G.~Knizhnik and A.~B.~Zamolodchikov,
Current algebra and Wess-Zumino model in two dimensions,
\href{https://doi.org/10.1016/0550-3213(84)90374-2}{Nucl.\ Phys.\ B {\bf 247}, 83 (1984)}.

\bibitem{LaMo1}
E.~Langmann and P.~Moosavi,
Construction by bosonization of a fermion-phonon model,
\href{https://doi.org/10.1063/1.4930299}{J.\ Math.\ Phys.\ {\bf 56}, 091902 (2015)}.

\bibitem{Voit}
J.~Voit,
One-dimensional Fermi liquids,
\href{https://doi.org/10.1088/0034-4885/58/9/002}{Rep.\ Prog.\ Phys.\ {\bf 58}, 977 (1995)}.

\bibitem{SCP}
H.~J.~Schulz, G.~Cuniberti, and P.~Pieri,
Fermi liquids and Luttinger liquids,
in {\it Field Theories for Low-Dimensional Condensed Matter Systems},
edited by G.~Morandi, P.~Sodano, A.~Tagliacozzo, and V.~Tognetti
\href{https://doi.org/10.1007/978-3-662-04273-1}{(Springer, Berlin, Heidelberg, 2000)}, p.~9.

\bibitem{CCaD}
J.~L.~Cardy, O.~A.~Castro-Alvaredo, and B.~Doyon, 
Form factors of branch-point twist fields in quantum integrable models and entanglement entropy,
\href{https://doi.org/10.1007/s10955-007-9422-x}{J.\ Stat.\ Phys.\ {\bf 130}, 129 (2008)}.

\bibitem{PrSe}
A.~Pressley and G.~Segal,
{\it Loop Groups}
(Oxford University Press, 1986).

\bibitem{KhWe}
B.~Khesin and R.~Wendt,
{\it The Geometry of Infinite-Dimensional Groups}
\href{https://doi.org/10.1007/978-3-540-77263-7}{(Springer, Berlin, Heidelberg, 2009)}.

\bibitem{GoWa1}
R.~Goodman and N.~R.~Wallach,
Structure and unitary cocycle representations of loop groups and the group of diffeomorphisms of the circle,
\href{https://doi.org/10.1515/crll.1984.347.69}{J.\ R.\ Angew.\ Math.\ {\bf 347}, 69 (1984)}.

\bibitem{GoWa2}
R.~Goodman and N.~R.~Wallach,
Projective unitary positive-energy representations of $\Diff(S^1)$,
\href{https://doi.org/10.1016/0022-1236(85)90090-4}{J.\ Func.\ Anal.\ {\bf 63}, 299 (1985)}.

\bibitem{ToLa}
V.~Toledano Laredo,
Integrating unitary representations of infinite-dimensional Lie groups,
\href{https://doi.org/10.1006/jfan.1998.3359}{J.\ Funct.\ Anal.\ {\bf 161}, 478 (1999)}.

\bibitem{Zamo}
A.~B.~Zamolodchikov,
Infinite additional symmetries in two-dimensional conformal quantum field theory,
\href{https://doi.org/10.1007/BF01036128}{Theor.\ Math.\ Phys.\ {\bf 65}, 1205 (1986)}.

\bibitem{Meln}
I.~V.~Melnikov,
{\it An Introduction to Two-Dimensional Quantum Field Theory with $(0,2)$ Supersymmetry}
\href{https://doi.org/10.1007/978-3-030-05085-6}{(Springer, Cham, 2019)}.

\bibitem{Spo2}
H.~Spohn,
{\it Large Scale Dynamics of Interacting Particles}
\href{https://doi.org/10.1007/978-3-642-84371-6}{(Spinger, Berlin, Heidelberg, 1991)}.

\bibitem{IlPe}
E.~Ilievski and T.~Prosen,
Thermodyamic bounds on Drude weights in terms of almost-conserved quantities,
\href{https://doi.org/10.1007/s00220-012-1599-4}{Commun.\ Math.\ Phys.\ {\bf 318} 809 (2013)}.

\bibitem{HeTe}
J.~Henheik and S.~Teufel,
Justifying Kubo's formula for gapped systems at zero temperature: a brief review and some new results,
\href{https://doi.org/10.1142/S0129055X20600041}{Rev.\ Math.\ Phys.\ {\bf 33}, 2060004 (2021)}.

\bibitem{MPT}
G.~Marcelli, G.~Panati, and S.~Teufel,
A new approach to transport coefficients in the quantum spin Hall effect,
\href{https://doi.org/10.1007/s00023-020-00974-6}{Ann.\ Henri Poincar{\'e} {\bf 22}, 1069 (2021)}.

\bibitem{Ish}
A.~Ishimaru,
{\it Wave Propagation and Scattering in Random Media}
\href{https://doi.org/10.1016/B978-0-12-374701-3.X5001-7}{(Academic Press, New York, 1978)}.

\bibitem{Doyon:Lecture_notes}
B.~Doyon,
Lecture notes on generalised hydrodynamics,
\href{https://doi.org/10.21468/SciPostPhysLectNotes.18}{SciPost Phys.\ Lect.\ Notes 18 (2020)}.

\bibitem{Moo:GHD-NLLM}
P.~Moosavi,
Emergence of generalized hydrodynamics in the non-local Luttinger model,
\href{https://doi.org/10.21468/SciPostPhys.9.3.037}{SciPost Phys.\ {\bf 9}, 037 (2020)}.

\end{thebibliography}
\end{document}